\documentclass[12 pt]{extarticle}
\usepackage{amsmath}
\usepackage{amssymb}
\usepackage{longtable}
\usepackage[a4paper, total={6in, 9in}]{geometry}
\usepackage{amsthm,url,tikz}
\usepackage[utf8]{inputenc}
\usepackage[english]{babel}
\usepackage{graphicx}
\usepackage{hyperref}
\usepackage{adjustbox}
\numberwithin{equation}{section}
\usepackage{makecell}
\usepackage{bm}
\usepackage{xcolor}

\newtheorem{theorem}{Theorem}[section]
\newtheorem{proposition}[theorem]{Proposition}
\newtheorem{corollary}[theorem]{Corollary}
\newtheorem{lemma}[theorem]{Lemma}

\theoremstyle{definition} 
\newtheorem{definition}{Definition}[section] 
\newtheorem{remark}{Remark}[section] 
\newtheorem{example}{Example}[section]

\usepackage{enumerate}
\usepackage[nottoc,notlot,notlof]{tocbibind} 

\title{Quasi-twisted codes and their connection with additive constacyclic codes over finite fields}
\author{Kanat Abdukhalikov\footnote{Email: abdukhalik@uaeu.ac.ae} \\
Department of Mathematical Sciences,\\
UAE University, PO Box 15551, Al Ain, UAE\\ \\
Gyanendra K. Verma\footnote{ gkvermaiitdmaths@gmail.com}\\
Department of Electrical Communication Engineering,\\ Indian Institute of Science, Bengaluru, Karnataka, India. }

\date{}

\begin{document}
	\maketitle
\begin{abstract} 
In this paper, we study quasi-twisted codes and their relationship with additive constacyclic codes through a polynomial-based approach. We first present a polynomial characterization of quasi-twisted codes over finite fields analogous to quasi-cyclic codes and determine  Euclidean, Hermitian, and symplectic duals of quasi-twisted codes with index $2$. Additionally, we provide necessary and sufficient conditions for the self-orthogonality of appropriate quasi-twisted codes. Next, we explore a one-to-one correspondence between quasi-twisted codes of length $lm$ with index $l$ over $\mathbb{F}_q$  and additive constacyclic codes of length $m$ over $\mathbb{F}_{q^l}$. We establish relationships between trace inner products in the additive setting and Euclidean, symplectic inner products in the quasi-twisted setting.  Using these relations and the correspondence, we determine the dual of additive constacyclic codes with respect to the trace inner products. As a consequence, we conclude that determining the trace Euclidean dual and trace Hermitian dual of an additive constacyclic code is equivalent to determining the Euclidean and symplectic dual of the corresponding quasi-twisted code.
\end{abstract}

\section{Introduction}
The characterization and construction of codes with excellent properties remain a central focus in the field of coding theory. Among the many classes of codes, those that support efficient encoding and decoding algorithms while demonstrating significant practical applications are regarded as good codes. Cyclic codes, in particular, stand out for their nice algebraic structure, which facilitates both efficient encoding and decoding processes. The origins of cyclic codes can be traced back to the pioneering work of Prange \cite{prange1957,prange1985}, where foundational connections between cyclic codes and algebra were first established. Over time, various extensions of cyclic codes have been extensively explored, including constacyclic codes, quasi-cyclic codes, and others. Among these, quasi-twisted codes are a significant generalization that includes cyclic, constacyclic, and quasi-cyclic codes. Moreover, quasi-twisted codes are also known to be asymptotically good (\cite{chepyzhov1993,wu2020,kasami1974}). Numerous record-breaking codes with good parameters are constructed within this class (e.g., see \cite{abdukhalikov20251,ackerman2011,aydin2017,daskalov2000,daskalov2003,qian2019}). Additionally, quasi-twisted codes with some orthogonality conditions have been utilized as classical ingredients in constructing quantum codes in \cite{ezerman2024,lv2020,saleh2024}.

A $\lambda$-quasi-twisted code of length $lm$ and index $l$ is an $\mathbb{F}_q[x]/\langle x^m-\lambda\rangle$ submodule of $(\mathbb{F}_q[x]/\langle x^m-\lambda\rangle)^l$. The structural characterization of quasi-twisted codes through their constituent codes has been extensively studied using the Chinese Remainder Theorem (CRT). This decomposition provides significant insights, enabling the derivation of the duals of quasi-twisted codes in terms of the duals of their constituent codes (see \cite{jia2011,lv2020,shi2016,ezerman2024}). Another promising approach to investigating quasi-twisted codes involves tuples of polynomials as a set of generators. These generators provide a direct and effective way to analyze the algebraic structure and properties of codes. 
 In \cite{lally2001}, the authors described the algebraic structure of quasi-cyclic codes in terms of polynomials using Gröbner bases without decomposing them into constituent codes. They further established a dimension formula for these codes, demonstrating the importance of this approach.  However, to the best of our knowledge, quasi-twisted codes have not been thoroughly explored using the polynomial approach.  Some advancements have been achieved in this direction, but they are primarily limited to one-generator quasi-twisted codes (see \cite{aydin2001, daskalov2003,saleh2024}).

Using the Chinese Remainder Theorem (CRT) decomposition approach, a quasi-twisted code over $\mathbb{F}_q$ is decomposed into a direct sum of constituent codes over high-degree extension fields of $\mathbb{F}_q$, determined by the factorization of $x^m-\lambda$ over $\mathbb{F}_q$. Consequently, the investigation of the structural properties, duality, and other characteristics of quasi-twisted codes generally involves computations over several high-degree extension fields of $\mathbb{F}_q$, which can become computationally demanding, particularly when the extension degrees are large. In contrast, the polynomial approach allows the entire analysis to be carried out over the base field $\mathbb{F}_q$, thereby avoiding the additional computational complexity associated with arithmetic in high-degree extension fields. Moreover, this framework yields compact generator and parity-check matrices that admit a natural implementation using hardware-friendly linear feedback shift registers (LFSRs), facilitating efficient encoding and syndrome computation. 
In addition, CRT approach requires the condition $\gcd(m,q)=1$, whereas the polynomial approach is applicable in many cases without this restriction.

Additive codes generalize linear codes by replacing the strict requirement of linearity with additivity, thereby forming a broader class of codes. Due to this generalization, computer-based searches have shown that additive codes often exhibit better parameters compared to linear codes of the same length and dimension. Moreover, there exist additive cyclic codes with specific parameters for which no linear cyclic codes exist (for example, see \cite{Verma2024}). Calderbank et al. \cite{calderbank1998} were the first to introduce additive cyclic codes over $\mathbb{F}_4$ and showed their connection with quantum codes. Huffman \cite{huffman2007,huffman2010} further investigated additive cyclic codes by employing decomposition into component codes. Self-orthogonal additive codes, defined with respect to trace inner products, can be used for constructing quantum stabilizer codes (for details, see \cite{ketkar2006}). These codes have gained considerable attention due to their advantages over linear codes and applications in quantum codes. Recently, several scholars studied additive cyclic codes using the polynomial approach with respect to trace inner products (see \cite{choi2023,reza2025,shi2022,shi2024,Verma2024}). Lally \cite{lally2003} derived that every quasi-cyclic code over finite fields corresponds to an additive cyclic code over some extension fields, where the degree of extension is the index of the quasi-cyclic code.

In this study, we explore quasi-twisted codes over finite fields and establish their relationship with additive constacyclic codes. Our contributions are as follows:

 \begin{enumerate}
     \item  We give a polynomial characterization of the structure of quasi-twisted codes without decomposing the code into its constituent codes (Section \ref{structure}).
     
     \item We determine the dual of quasi-twisted codes with index $2$ with respect to Euclidean, Hermitian, and symplectic inner products. Additionally, we establish the necessary and sufficient conditions for self-orthogonality of appropriate quasi-twisted codes with respect to these inner products (Section \ref{dualcodes}, \ref{dualcodess}, \ref{dualcodesh}, and \ref{selfdualcodes}).

     \item  We explore a correspondence between quasi-twisted codes and additive constacyclic codes. We determine the dual of additive constacyclic codes with respect to trace Euclidean, trace Hermitian, and symplectic inner products from duals of the corresponding quasi-twisted codes with respect to Euclidean, Hermitian, and symplectic inner products by selecting suitable bases of the extension field. As a consequence, our findings generalize the results obtained in \cite{reza2025,guneri2018,shi2022,Verma2024} (Section \ref{qtl2ca} and \ref{qttoconstal}).
 \end{enumerate}

\section{Preliminaries}\label{pre}
Let $\mathbb{F}_q$ be a finite field with $q$ elements, where $q$ is a prime power. A linear code of length $n$ over $\mathbb{F}_q$ is a linear subspace of $\mathbb{F}_q^n$.
 Let $\lambda \in \mathbb{F}_q\setminus \{0\}$. We define the cyclic $\lambda$-shift operator as 
 $$T_{\lambda}(x_0,x_1,\dots,x_{n-1})=(\lambda x_{n-1},x_0,\dots,x_{n-2}).$$
 A linear code is said to be $\lambda$-quasi-twisted code of length $n$ and index $l$  if $T_{\lambda}^l(c)\in C$ for all $c\in C$. If $\lambda=1$, then $C$ is called a quasi-cyclic code of index $l$. When $l=1$, then $C$ is called $\lambda$-constacyclic code, and when $l=1$, $\lambda=1$, $C$ is called a cyclic code. One can assume that $n=ml$. 

 Denote $R=\mathbb{F}_q[x]/\langle x^m-\lambda\rangle$. Note that $\lambda$-constacyclic codes of length $m$ are ideals of $R$. Let $C$ be a $\lambda$-quasi-twisted code of length $n=ml$ and index $l$ over $\mathbb{F}_q$. It is well known that $C$ can be identified as a $R$-submodule of $R^l$ (for instance, see \cite{jia2011}). The correspondence is given by
 $$(c_{0,0},c_{0,1},\dots,c_{0,l-1},c_{1,0},\dots c_{1,l-1},\dots,c_{m-1,0},\dots,c_{m-1,l-1})\mapsto (c_0(x),c_1(x),\dots c_{l-1}(x))\in R^l,$$
 where $c_i(x)=c_{0,i}+c_{1,i}x+\dots+c_{m-1,i}x^{m-1}\in R$.

\section{Structure of quasi-twisted codes of index 2}
\label{structure}
The structural characterization of quasi-twisted codes via their constituent codes has been extensively investigated using the Chinese Remainder Theorem (CRT). In \cite{lally2001}, Lally and Fitzpatrick described the algebraic structure of quasi-cyclic codes in terms of polynomials through Gröbner bases, without resorting to a decomposition into constituent codes. In this section, we explore the structural properties of quasi-twisted codes of index two using a polynomial-based approach.  In the following theorem, we characterize the polynomial structure of quasi-twisted codes of index $2$. The proof follows the approach of \cite{lally2001}; for completeness, we provide a brief outline of the argument.

\begin{theorem}\label{qtcode}
   1) Let $C$ be a $\lambda$-quasi-twisted code of length $2m$ and index $2$. Then there exist $g_1=(g_{11}(x),g_{12}(x)), g_2=(0,g_{22}(x))\in (\mathbb{F}_q[x])^2$ such that $C$ is generated by $g_1$, $g_2$ and the following holds:
    \begin{equation}\label{eqqt}
    \begin{split}
       g_{11}(x),g_{22}(x)\text{ divide } (x^m-\lambda), \\
       \deg g_{12}(x)< \deg g_{22}(x),\\
       g_{11}(x)g_{22}(x)\text{ divides } (x^m-\lambda)g_{12}(x).
    \end{split}        
    \end{equation} 
Moreover, in this case 
$\dim C = 2m- \deg g_{11}(x) - \deg g_{22}(x)$.      

2) Let code $C$ be generated by elements $g_1=\big( g_{11}(x),g_{12}(x)\big)$ and $g_2=\big( 0,g_{22}(x)\big)$, 
and let $C'$ be generated by elements $g'_1=\big( g'_{11}(x),g'_{12}(x)\big)$ and $g'_2=\big( 0,g'_{22}(x)\big)$, 
both satisfying Eq. \ref{eqqt}.  Assume that  $g_{11}(x)$, $g_{22}(x)$, $g'_{11}(x)$, $g'_{22}(x)$ are monic polynomials. 
Then $C=C'$ if and only if $g_{11}(x) = g'_{11}(x)$, $g_{22}(x) = g'_{22}(x)$, and $g_{12}(x) = g'_{12}(x)$.  
\end{theorem}

\begin{proof}
 1) As $C$ is  a $\lambda$-quasi-twisted code of index $2$, that is, $C$ is a submodule of $R^2$. Define a map 
$\phi: \mathbb{F}_q[x]^2\to R^2$, $\phi(p(x),q(x))=(p(x)\pmod{x^m-\lambda},q(x)\pmod{x^m-\lambda})$. Then $\phi$ is a module homomorphism. Let $\tilde{C}$  is the pre-image of $C$ under $\phi$. Then $\tilde{C}$ is a submodule of $\mathbb{F}_q[x]^2$ containing the submodule $\tilde{K}$ generated by $\{(x^m-\lambda,0), (0,x^m-\lambda)\}$. Since $\tilde{C} $ is a submodule of a finitely generated free module over a principal ideal domain $\mathbb{F}_q[x]$ containing $\tilde{K}$, it has a generating set of the form $\{(x^m-\lambda,0), (0,x^m-\lambda), r_1,r_2,\dots,r_t\}$, where $r_i=(r_{i1}(x),r_{i2}(x))$, that is, the rows of the matrix 
\begin{align*}
   M= \begin{bmatrix}
      r_{11}(x)&r_{12}(x)\\  
      r_{21}(x)&r_{22}(x)\\  
      \vdots & \vdots\\
      r_{t1}(x)&r_{t2}(x)\\ 
      x^m-\lambda &0\\
      0 & x^m-\lambda
    \end{bmatrix}
\end{align*} 
generate $\tilde{C}$. Applying elementary row operations in $\mathbb{F}_q[x]$, we get
\begin{align*}
   \tilde{M}= \begin{bmatrix}
      g_{11}(x)&g_{12}(x)\\  
      0(x)&g_{22}(x)\\  
      0 &0 \\
      \vdots & \vdots\\
      0&0\\ 
      0&0\\
      0 & 0
    \end{bmatrix},
\end{align*} where $g_{11}(x)$ and $g_{22}(x)$ divide $x^m-\lambda$ and $\deg g_{12}(x)<\deg g_{22}(x)$. Therefore, $C$ is generated by $(g_{11}(x),g_{12}(x))$ and $(0,g_{22}(x))$. The rest of the properties of these generators follow from \cite[Theorem 2.1]{lally2001}.
In particular, $\dim C = 2m- \deg g_{11}(x) - \deg g_{22}(x)$. 
Existence of such generators are equivalent \cite{lally2001} to the existence of a $2\times 2$ polynomial matrix 
$(a_{ij})$ such that 
$$
\left(
\begin{array}{cc}
 a_{11} &  a_{12}  \\
 a_{21} &  a_{22} 
\end{array}
\right)
\left(
\begin{array}{cc}
 g_{11} &  g_{12}  \\
 0 &  g_{22} 
\end{array}
\right) = 
\left(
\begin{array}{cc}
 x^m-\lambda&  0  \\
 0 &  x^m-\lambda
\end{array}
\right).
$$
Then $a_{21}=0$,  $a_{11}=\frac{x^m-\lambda}{g_{11}(x)}$,  $a_{22}=\frac{x^m-\lambda}{g_{22}(x)}$,  and 
$\frac{x^m-\lambda}{g_{11}(x)} g_{12} + a_{12} g_{22} =0$, which implies  $g_{11}(x)g_{22}(x) \mid  (x^m-\lambda) g_{12}(x)$. 

2) Consider the projection mapping $P : C \rightarrow R$, $P\big(a(x),b(x)\big) = a(x)$.  
Then $\ker P = \langle \big(0,g_{22}(x)\big) \rangle$  and ${\rm Im} \ P = \langle g_{11}(x) \rangle$. 
If $C$ is generated by elements $(0,g_{22}(x))$ and $\big( g_{11}(x),g'_{12}(x)\big)$, then $g_{12}(x) - g'_{12}(x)$ 
is divisible by   $g_{22}(x)$. 
\end{proof}

\begin{remark}\label{remgcd1}
    If $\gcd(m,q)=1$ then the condition $g_{11}(x)g_{22}(x)\text{ divides } (x^m-\lambda)g_{12}(x)$ in Eq. \ref{eqqt} is equivalent to the condition that $\gcd (g_{11}(x),g_{22}(x))$ divides $g_{12}(x)$, since $x^m-\lambda$ has distinct irreducible factors. 
\end{remark}

 \begin{theorem}\label{2to1}
Let $C$ be a $\lambda$-quasi-twisted code over $\mathbb{F}_q$ of length $2m$ and index $2$, generated by two elements $g_1=(g_{11}(x),g_{12}(x))$ and $g_2=(0,g_{22}(x))$ satisfying Eq. \ref{eqqt}. If $C$ is generated by one element, then $g_{11}(x)g_{22}(x)\equiv 0\pmod{x^m-\lambda}$. The converse holds whenever $\gcd(m,q)=1$, that is, if $\gcd(m,q)=1$ and $g_{11}(x)g_{22}(x)\equiv 0\pmod{x^m-\lambda}$ then $C$ is generated by one element.     
 \end{theorem}
 \begin{proof}
     Suppose $C$ is generated by one element $h=(h_1(x),h_2(x))\in R^2$. Then
     $(g_{11}(x),g_{12}(x))=a(x)(h_1(x),h_2(x))$ and $(0,g_{22}(x))=b(x)(h_1(x),h_2(x))$ in $R^2$ for some $a(x),b(x)\in R$.
     Equivalently,
     \begin{equation*}
         \begin{split}
             g_{11}(x)\equiv a(x)h_1(x) \pmod{x^m-\lambda},\\
             g_{12}(x)\equiv a(x)h_2(x) \pmod{x^m-\lambda},\\
             0\equiv b(x)h_1(x)\pmod{x^m-\lambda},\\
             g_{22}(x)\equiv b(x)h_2(x) \pmod{x^m-\lambda}.
         \end{split}
     \end{equation*}
 Then $b(x)g_{11}(x)\equiv a(x)b(x)h_1(x)\equiv 0 \pmod{x^m-\lambda}$.  Consequently, $g_{11}(x)g_{22}(x)\equiv b(x)g_{11}(x)h_2(x)\equiv 0 \pmod{x^m-\lambda}$.

 Conversely, assume that $\gcd(m,q)=1$ and $g_{11}(x)g_{22}(x)\equiv 0 \pmod{x^m-\lambda}$. Let $\gcd(g_{11}(x),g_{22}(x))=g(x)$. By Remark \ref{remgcd1}, $g(x)|g_{12}(x)$. Suppose $g_{11}(x)=g(x)g_{11}'(x)$, $g_{22}(x)=g(x)g_{22}'(x)$ and $g_{12}(x)=g(x)g_{12}'(x)$ with $\gcd(g_{11}'(x),g_{22}'(x))=1$.
 Then, $g_{11}(x)g_{22}(x)\equiv 0 \pmod{x^m-\lambda}$ implies that \begin{align*}
     g_{11}(x)g_{22}(x)= (x^m-\lambda)k(x) \text{ for some } k(x)\in \mathbb{F}_q[x]\\
     \implies g(x)^2g_{11}'(x) g_{22}'(x)=(x^m-\lambda)k(x)
 \end{align*}
 Since $\gcd(m,q)=1,$  $x^m-\lambda$ is square-free. Also, $g(x) \mid (x^m - \lambda)$, $x^m - \lambda$ cannot contain $g(x)^2$ as a factor.  Thus, $g(x)$ must divide $k(x)$, that is,  $k(x)=g(x)k'(x)$  for some $k'(x)\in \mathbb{F}_q[x]$. Thus,
 \begin{align*}
     g(x)g_{11}'(x) g(x) g_{22}'(x)=(x^m-\lambda)g(x)k'(x),\\
      g(x)g_{11}'(x) g_{22}'(x)=(x^m-\lambda)k'(x).
 \end{align*}
 Since relatively coprime polynomials $g(x),g_{11}'(x)$ and $g_{22}'(x)$ are divides $x^m-\lambda$, write $x^m-\lambda=h(x)g(x)g_{11}'(x) g_{22}'(x)$ for some polynomial $h(x)\in \mathbb{F}_q[x]$. It follows that \begin{align*}
    g(x)g_{11}'(x) g_{22}'(x) =h(x)g(x)g_{11}'(x) g_{22}'(x)k'(x), \text{ i.e., } h(x)k'(x)=1. 
 \end{align*}
 Without loss of generality, we can assume $g(x),g_{11}'(x)$ and $g_{22}'(x)$ are monic polynomials, that implies $h(x)$ is a monic polynomial, and consequently, we get $k'(x)=1$, that is, $g(x)g_{11}'(x) g(x) g_{22}'(x)=(x^m-\lambda)$ . Therefore, we have 
 \begin{equation*}
     \begin{split}
  (x^m-\lambda)g(x)=g_{11}(x)g_{22}(x),\   x^m-\lambda=g(x)g_{11}'(x)g_{22}'(x), \\
      g_{22}(x)=\frac{x^m-\lambda}{g_{11}'(x)},  \text{ and }\frac{x^m-\lambda}{g_{11}(x)}g_{12}(x)=g_{22}(x)g_{12}'(x).
     \end{split}
 \end{equation*}
Also, since $\gcd(g_{22}'(x),\frac{x^m-\lambda}{g_{22}(x)})=1$, there exists $f(x)$ such that 
         $$f(x)g_{22}'(x)\equiv 1  \left(\bmod{\frac{x^m-\lambda}{g_{22}(x)}}\right).$$
   Equivalently, $$f(x)g_{22}'(x)g_{22}(x)\equiv g_{22}(x) \pmod{x^m-\lambda}.$$
         Consequently,  
      $$-f(x)g_{22}'(x)g_{22}(x)g_{12}'(x)\equiv -g_{12}'(x)g_{22}(x) \pmod{x^m-\lambda}.$$ 
      Let $b(x)=f(x)-f(x)g_{12}'(x)$. Consider a one-generator code $C'$ generated by $ A=(g_{11}(x), g_{12}(x)+b(x)g_{22}(x))$. We claim that $C=C'$, that is, $C$ is generated by one element. Observe that $A=(g_{11}(x),g_{12}(x))+b(x)(0,g_{22}(x))\in C$. Thus, $C'\subseteq C$.\\
 Note that \begin{equation*}
     \begin{split}
         \frac{x^m-\lambda}{g_{11}(x)}A&\equiv \left (0,\ \frac{x^m-\lambda}{g_{11}(x)}g_{12}(x)+b(x)\frac{x^m-\lambda}{g_{11}(x)}g_{22}(x)\right )\\
         \\
         &\equiv (0,\ g_{22}(x)g_{12}'(x)+[f(x)-f(x)g_{12}'(x)]g_{22}(x)g_{22}'(x))\\
         \\
         &\equiv (0,\ g_{22}(x)g_{12}'(x)+f(x)g_{22}(x)g_{22}'(x)-f(x)g_{12}'(x)g_{22}(x)g_{22}'(x))\\
         \\
         &\equiv (0,\  g_{22}(x)g_{12}'(x)+g_{22}(x)-g_{12}'(x)g_{22}(x))\\
         \\
         &\equiv (0,\ g_{22}(x)) \pmod{x^m-\lambda}.
     \end{split}
 \end{equation*} 
 This shows that $g_2=(0,g_{22}(x))\in C'$. Moreover, \begin{align*}
     g_1=(g_{11}(x),g_{12}(x))=& (g_{11}(x),g_{12}(x)+b(x)g_{22}(x))-b(x)(0,g_{22}(x))\\
     =&A-b(x)g_2\in C'.
 \end{align*}
 Hence, $C\subseteq C'$. This completes the proof.
\end{proof}
 

     
In the following example, we show that the condition $g_{11}(x)g_{22}(x)\equiv 0\pmod{x^m-\lambda}$ in Theorem \ref{2to1} is not sufficient in case $\gcd(m,q)\ne 1$.

\begin{example}
 Let $\mathbb{F}_q=\mathbb{F}_2$, $\lambda=1$ and $m=6$. Then $x^6-1=(x+1)^2(x^2+x+1)^2$. Let $g_{11}(x)=(x+1)(x^2+x+1)$, $g_{12}(x)=x(x+1)(x^2+x+1)$ and $g_{22}(x)=(x+1)(x^2+x+1)^2$. Consider a quasi-cyclic code $C$ of length $2m$ and index $2$ generated by $g_1=(g_{11}(x),g_{12}(x))$ and $g_2=(0,g_{22}(x))$. Note that $g_{11}(x)g_{22}(x)\equiv 0\pmod{x^m-1}$, but  $C$ cannot be generated by one element.   
\end{example}

\section{Euclidean dual of quasi-twisted codes}\label{dualcodes}
In this section, we determine the Euclidean duals of quasi-twisted codes over $\mathbb{F}_q$ within the polynomial framework. We first establish several properties of the Euclidean inner product in terms of polynomial representations. These properties are then utilized to characterize the Euclidean duals of both one-generator and two-generator quasi-twisted codes. In particular, the polynomial formulation enables the dual codes to be determined directly from the corresponding generator polynomials, providing an explicit description of their algebraic structure. 

Denote $R^*=\mathbb{F}_q[x]/\langle (x^m-\lambda^{-1})\rangle$. For any $a(x)=a_0+a_1x+\dots+a_{m-1}x^{m-1},b(x)=b_0+b_1x+\dots+b_{m-1}x^{m-1}$, the Euclidean inner product is defined as
$$\langle a(x),b(x)\rangle_e=a_0b_0+a_1b_1+\dots+a_{m-1}b_{m-1}.$$
If $a(x)$ has degree greater than or equal to $m$, then we reduce it modulo $x^m-\lambda$. 
If $b(x)$ has degree greater than or equal to $m$, then we reduce it modulo $x^m-\lambda^{-1}$.
Hence, the left component is considered as an element of $R$, and the right component is considered as an element of $R^*$.  Therefore, bilinear map $\langle \cdot, \cdot\rangle_e$ can be considered as a pairing \cite{Lang} of $R \times R^*$ onto $\mathbb{F}_q$.  The inner product naturally extends componentwise over $R^l$ and $R^*{^l}$. 
\begin{lemma}
    The pairing $\langle \cdot, \cdot\rangle_e: R\times R^*\to \mathbb{F}_q$ is nondegenerate.
\end{lemma}
\begin{proof}
 Let $h(x)=h_0+h_1x+h_2x^2+\dots+h_{m-1}x^{m-1}\in R^{*}$. 
For $i=0,1,2,\dots,m-1$, $x^i\in R$ and  
$$\langle x^i,h(x)\rangle_e=0\implies h_i=0.$$
Therefore, $\langle a(x),h(x)\rangle_e=0$ for all $a(x)\in R$ implies that $h(x)=0$ in $R^*$. 
\end{proof}

\begin{definition}
    Let $t(x)=t_0+t_1x+\dots+t_dx^d\in \mathbb{F}_{q}[x]$ be a nonzero polynomial, then the reciprocal of $t(x)$ is denoted by $t^*(x)$ and is defined as $t^*(x)=x^{\deg t(x)}t(x^{-1})$. We define $t^*(x)=0$ whenever $t(x)=0$.
\end{definition}
Note that $(x^m-\lambda)^*=(1-\lambda x^m)=-\lambda (x^m-\lambda^{-1})$. It is easy to see that $t^*(x)$ divides $(x^m-\lambda^{-1})$ whenever $t(x)$ divides $(x^m-\lambda)$.
\begin{definition}
    For a polynomial $t(x)=t_0+t_1x+\dots+t_mx^m$, we define $ t^{\circ}(x)=\lambda x^mt(x^{-1})$.
\end{definition}
 Now, we describe some properties of the polynomial $t^{\circ}(x)$ and its connection with $t^*(x)$.

\begin{lemma}\label{properties}
\begin{enumerate}[(1)]
    \item  If $t(x)=t_0+t_1x+\dots+t_dx^d$ with  $d=\deg t(x) \leq m$, then $t^\circ(x)=\lambda x^{m-\deg t(x)} t^*(x)$.
 \item If $t(x)=t_0+t_1x+\dots+t_mx^m$ then $$t^\circ(x) \equiv (t_0+\lambda t_m)+ \lambda (t_{m-1}x+\dots t_2x^{m-2}+t_1x^{m-1})\pmod {x^m-\lambda^{-1}}.$$
   \end{enumerate}
\end{lemma}
\begin{proof}
\begin{enumerate}[(1)]
    \item  $t^{\circ}(x)=\lambda x^m t(1/x)=\lambda x^{m-d}x^dt(1/x)=\lambda x^{m-d}t^*(x)$.
 \item We have \begin{equation*}
 \begin{split}
   t^\circ(x)&=\lambda x^mt(1/x)=\lambda (t_0x^m+t_1x^{m-1}+\dots+t_{m-2}x^2+t_{m-1}x+t_m)\\
 &\equiv (t_0+\lambda t_m)+ \lambda (t_{m-1}x+\dots t_2x^{m-2}+t_1x^{m-1})\pmod {x^m-\lambda^{-1}},
 \end{split}
 \end{equation*}
      \end{enumerate}
which completes the proof.
\end{proof}

\begin{lemma}\label{hat}
 Let $t(x)\in \mathbb{F}_q[x]$ with $\text{deg}\ t(x)\leq m-1$. Then 
 \begin{equation}
     \langle a(x)t(x)),b(x)\rangle_e=\langle a(x), b(x)t^{\circ}(x)\rangle_e
 \end{equation} for all $a(x),b(x)\in \mathbb{F}_q[x]$, where $a(x), a(x)t(x)$ and $b(x), b(x)t^{\circ}(x)$ are considered as an element in $R$ and $R^*$, respectively.
\end{lemma}

\begin{proof}
 Observe that $(t(x)+s(x))^{\circ}=t^{\circ}(x)+s^{\circ}(x)$. Indeed, let $t(x)=t_0+t_1x+\dots+t_{m-1}x^{m-1}$ and $s(x)=s_0+s_1x+\dots+s_{m-1}x^{m-1}\in \mathbb{F}_q[x]$. Then \begin{equation*}
     \begin{split}
        t^{\circ}(x)+s^{\circ}(x)&=\lambda x^m(t_0+t_1/x+\dots+t_{m-1}/x^{m-1})+\lambda x^m(s_0+s_1/x+\dots+s_{m-1}/x^{m-1})\\
        &=\lambda x^m((t_0+s_0)+(t_1+s_1)/x+\dots+(t_{m-1}+s_{m-1})/x^{m-1})\\
        &=\lambda x^m (t+s)(1/x)=(t(x)+s(x))^{\circ}.
     \end{split}
 \end{equation*} 
 Let $a(x) \equiv a_0+a_1x+\dots a_{m-1}x^{m-1}\pmod{x^m-\lambda}$ and $b(x) \equiv b_0+b_1x+\dots b_{m-1}x^{m-1}\pmod{x^m-\lambda^{-1}}$. We prove the result for monomials $t(x)=x^i$ for $0\leq i\leq m-1$. Consequently, the result will follow for any polynomial $t(x)$ of degree less than or equal to $m-1$. For $0\leq i\leq m-1$, $t(x)=x^i$, we have 
    \begin{equation*}
        \begin{split}
            a(x)x^i&=a_0x^i+a_1x^{i+1}+\dots+a_{m-i-1}x^{m-1}+a_{m-i}x^m+\dots+a_{m-1}x^{m-1+i}\\
            &\equiv \lambda (a_{m-i}+a_{m-i+1}x+\dots+a_{m-1}x^{i-1})+a_0x^i+\dots+a_{m-i-1}x^{m-1}\pmod{x^m-\lambda}.
        \end{split}
    \end{equation*}
    Thus, \begin{equation}\label{at}
        \langle a(x)x^i, b(x)\rangle_e=\lambda (a_{m-i}b_0+a_{m-i+1}b_1+\dots+a_{m-1}b_{i-1})+a_0b_i+\dots+a_{m-i-1}b_{m-1}.
    \end{equation}

     Also, $t^{\circ}(x)=(\lambda x^m)t(1/x)=\lambda x^{m-i}$ and
     \begin{equation*}
        \begin{split}
            b(x)\lambda x^{m-i}&=\lambda[b_0x^{m-i}+b_1x^{m-i+1}+\dots+b_{i-1}x^{m-1}+b_ix^m+\dots+b_{m-1}x^{m-1+m-i}]\\
            \equiv &[b_i+b_{i+1}x+\dots b_{m-1}x^{m-i-1}]+\lambda[b_0x^{m-i}+b_1x^{m-i+1}+\dots+ b_{i-1}x^{m-1}] \pmod{x^m-\lambda^{-1}}.
        \end{split}
    \end{equation*}
    Thus,
\begin{equation}\label{bt}
    \langle a(x), b(x)\lambda x^{m-i}\rangle_e= a_0b_i+\dots+a_{m-i-1}b_{m-1}+\lambda (a_{m-i}b_0+a_{m-i+1}b_1+\dots+a_{m-1}b_{i-1}).
\end{equation}
Hence, comparing Eqs. \ref{at} and \ref{bt}, we have $\langle a(x)t(x),b(x)\rangle_e=\langle a(x), (b(x)t^{\circ}(x)\rangle_e$ for $t(x)=x^i,\ 0\leq i\leq m-1$. This completes the proof. \end{proof}

If $C$ is a code in $R^l$, then its Euclidean  dual is defined as 
$$C^{\perp_e}=\{d\in (R^*)^l | \ \langle c, d \rangle_e=0 \ \forall c\in C\}.$$

\begin{proposition}\cite[Proposition 1]{jia2011}\label{ldual}
  If $C$ is a $\lambda$-quasi-twisted code with index $l$ then $C^{\perp_e}$ is a $\lambda^{-1}$-quasi-twisted code with index $l$.  
\end{proposition}

Next, we determine the tuples of polynomials that generate the Euclidean duals of quasi-twisted codes with index 2. This characterization offers an explicit algebraic framework for understanding the structure of the dual codes.

\begin{theorem}\label{thmedual}
    Let $C$ be a $\lambda$-quasi-twisted code over $\mathbb{F}_q$ of length $2m$ and index $2$, generated by two elements $g_1=(g_{11}(x),g_{12}(x))$ and $g_2=(0,g_{22}(x))$ satisfying Eq. \ref{eqqt}.    Then the Euclidean dual $C^{\perp_e}$  of $C$ is a $\lambda^{-1}$-quasi-twisted code of length $2m$ and index $2$, and generated by elements $$\left (\frac{x^m-\lambda^{-1}}{g^*_{11}(x)},0\right) \text{ and } \left(\frac{-\lambda x^{m+\deg g_{11}(x)-\deg g_{12}(x)}(x^m-\lambda^{-1})g^*_{12}(x)}{g^*_{11}(x)g^*_{22}(x)},\frac{x^m-\lambda^{-1}}{g^*_{22}(x)}\right).$$
\end{theorem}
 \begin{proof}
By Proposition \ref{ldual}, $C^{\perp_e}$ is a $\lambda^{-1}$-quasi-twisted code of length $2m$ and index $2$. Assume $C'\subseteq R^*$ is a $\lambda^{-1}$-quasi-twisted code generated by elements $f_1=(f_{11}(x),0)$ and $f_2=(f_{21}(x),f_{22}(x))$, where  
$$f_{11}(x)=\frac{x^m-\lambda^{-1}}{g^*_{11}(x)},  f_{22}(x)=\frac{x^m-\lambda^{-1}}{g^*_{22}(x)}$$
and  $$f_{21}(x)=\frac{-\lambda x^{m+\deg g_{11}(x)-\deg g_{12}(x)}(x^m-\lambda^{-1})g^*_{12}(x)}{g^*_{11}(x)g^*_{22}(x)}.$$ 
 Note that  $g_{11}(x)g_{22}(x)$ divides  $(x^m-\lambda)g_{12}(x)$ (by Eq. \ref{eqqt}). Consequently, $g^*_{11}(x)g^*_{22}(x)$ divides $(x^m-\lambda^{-1})g^*_{12}(x)$. Hence  $f_{21}(x)$ is a polynomial in $\mathbb{F}_q[x]$.
We show that $C'=C^{\perp_e}$. Observe that \begin{equation*}
    \begin{split}
        f_{11}(x) g_{11}^{\circ}(x)=\frac{x^m-\lambda^{-1}}{g^*_{11}(x)} \cdot \lambda x^{m-\deg g_{11}(x)}g^*_{11}(x)=\lambda x^{m-\deg g_{11}(x)}(x^m-\lambda^{-1}) =0 \text{ in } R^*\\
        f_{22}(x) g_{22}^{\circ}(x)=\frac{x^m-\lambda^{-1}}{g^*_{22}(x)} \cdot \lambda x^{m-\deg g_{22}(x)}g^*_{22}(x)=\lambda x^{m-\deg g_{22}(x)}(x^m-\lambda^{-1}) =0 \text{ in } R^*\\
        f_{21}(x)g^{\circ}_{11}(x)+f_{22}(x)g_{12}^{\circ}(x)=\frac{-\lambda x^{m+\deg g_{11}(x)-\deg g_{12}(x)}(x^m-\lambda^{-1})g^*_{12}(x)}{g^*_{11}(x)g^*_{22}(x)}\lambda x^{m-\deg g_{11}(x)}g^*_{11}(x)\\
        + \frac{x^m-\lambda^{-1}}{g^*_{22}(x)} \lambda x^{m-\deg g_{12}(x)} g^*_{12}(x)\\
        =\frac{-\lambda^2 x^{2m-\deg g_{12}(x)}(x^m-\lambda^{-1})g^*_{12}(x)}{g^*_{22}(x)}
        + \frac{x^m-\lambda^{-1}}{g^*_{22}(x)} \lambda x^{m-\deg g_{12}(x)} g^*_{12}(x)\\
        =\frac{-\lambda x^{m-\deg g_{12}(x)}(x^m-\lambda^{-1})g^*_{12}(x)}{g^*_{22}(x)}
        + \frac{x^m-\lambda^{-1}}{g^*_{22}(x)} \lambda x^{m-\deg g_{12}(x)} g^*_{12}(x)=0 \text{ in } R^* \\
    \end{split}
\end{equation*} 
since  $x^m=\lambda^{-1} \text{ in } R^*$. 
Thus, $\langle g_2, f_1\rangle_e=0$ (trivially) and  using Lemma \ref{hat}, we have 
\begin{equation*}
\begin{split}
        \langle g_1, f_1\rangle_e&= \langle g_{11}(x), f_{11}(x)\rangle_e= \langle 1, f_{11}(x) g_{11}^{\circ}(x)\rangle_e=0,\\
        \langle g_2, f_2\rangle_e&= \langle g_{22}(x), f_{22}(x)\rangle_e= \langle 1, f_{22}(x) g_{22}^{\circ}(x)\rangle_e=0,\\
       \langle g_1,f_2\rangle_e&= \langle  g_{11}(x),f_{21}(x)\rangle_e+\langle  g_{12}(x),f_{22}(x)\rangle_e\\
       &=\langle 1, f_{21}(x)g^{\circ}_{11}(x)+f_{22}(x)g_{12}^{\circ}(x)\rangle_e=0.
        \end{split}
    \end{equation*}

    Hence, $C'\subseteq C^{\perp_e}$. Observe that $C'$ satisfies conditions in $R^*$ similar to Eq. \ref{eqqt}. 
By uniqueness, we have  $\dim (C')=2m-\deg f_{11}(x)-\deg f_{22}(x)$. Also,  since  $\deg g_{11}(x)=\deg g^*_{11}(x)$  and  $\deg g_{22}(x)=\deg g^*_{22}(x))$ therefore
    \begin{equation*}
        \begin{split}
            \dim(C)+\dim(C')&=2m-\deg g_{11}(x)-\deg g_{22}(x)+2m-\deg f_{11}(x)-\deg f_{22}(x)\\
            &=2m-\deg g_{11}(x)-\deg g_{22}(x)\\
            &+2m- (m-\deg g^*_{11}(x))- (m-\deg g^*_{22}(x))\\
            &=2m\  
        \end{split}
    \end{equation*}
Thus $C'=C^{\perp_e}$. This completes the proof.
\end{proof}

 \begin{theorem}
   Let $\gcd(m,q)=1$. Let $C$ be a $\lambda$-quasi-twisted code over $\mathbb{F}_q$ of length $2m$ and index $2$, generated by two elements $g_1=(g_{11}(x),g_{12}(x))$ and $g_2=(0,g_{22}(x))$ satisfying Eq. \ref{eqqt}. Then $C^{\perp_e}$ is generated by one element if and only if $g_{11}(x)g_{22}(x)$ divides $x^m-\lambda$.  
 \end{theorem}
 \begin{proof}
   The proof follows from Theorems \ref{2to1} and \ref{thmedual}.
 \end{proof}

It is well known that if $C$ is a one-generator quasi-cyclic code of length $2m$ and index $2$, generated by $(g_{11}(x),g_{12}(x))$ such that $g_{11}(x)\mid x^m-1$, then $\dim (C)=m-\deg g(x)$, where $g(x)=\gcd(g_{11}(x),g_{12}(x))$ (for instance, see \cite[Lemma 1]{Seguin2004}). Using similar arguments, we have the following.

\begin{lemma}
  Let $\gcd(m,q)=1$.  If $C$  is a one-generator $\lambda$-quasi-twisted code of length $2m$ and index $2$, generated by $(g_{11}(x),g_{12}(x))$ such that $g_{11}(x)\mid x^m-\lambda$, then $\dim (C)=m-\deg g(x)$, where $g(x)=\gcd(g_{11}(x),g_{12}(x))$.
\end{lemma}

\begin{theorem}\label{gen1edaul}
Let $\gcd(m,q)=1$. Let $C$ be one-generator $\lambda$-quasi-twisted code of length $2m$ and index $2$, generated by $(g_{11}(x),g_{12}(x))$ with $g_{11}(x)\mid x^m-\lambda$. Let $g(x)=\gcd(g_{11}(x),g_{12}(x))$, $g_{11}(x)=g(x)g_{11}'(x)$ and $g_{12}(x)=g(x)g_{12}'(x)$. Then Euclidean dual of $C$ is a two-generator $\lambda^{-1}$-quasi-twisted code of length $2m$ and index $2$, generated by 
 $$\left (\frac{x^m-\lambda^{-1}}{g_{11}^*(x)},0\right ) \text{ and } (-\lambda x^{m+\deg g_{11}'(x)-\deg g_{12}'(x)}g_{12}'^*(x),g_{11}'^*(x)).$$
 
\end{theorem}

\begin{proof}
Let $D$ be a $\lambda^{-1}$-quasi-twisted code of length $2m$ and index $2$, generated by 
 $$\left (\frac{x^m-\lambda^{-1}}{g_{11}^*(x)},0\right ) \text{ and } (-\lambda x^{m+\deg g_{11}'(x)-\deg g_{12}'(x)}g_{12}'^*(x),g_{11}'^*(x)).$$
Since $\gcd((x^m-\lambda^{-1})/g_{11}^*(x), g_{11}'^*(x))=1 $, therefore, by Theorem \ref{qtcode}, 
$$\dim (D)=2m-(m-\deg g_{11}(x))-\deg g_{11}'(x)=m+\deg g(x).$$
We show that $D=C^{\perp_e}.$ First, $\dim(C^{\perp_e})=2m-\dim(C)=2m-(m-\deg g(x))=m+\deg g(x)$, that is, $\dim (C^{\perp_e})=\dim(D).$
Also, 
\begin{equation*}
    \begin{split}
        \langle (g_{11}(x),g_{12}(x)), ((x^m-\lambda^{-1})/g_{11}^*(x), 0)\rangle_e=0
    \end{split}
\end{equation*}
and 
\begin{equation*}
    \begin{split}
     &\langle (g_{11}(x),g_{12}(x)), (-\lambda x^{m+\deg g_{11}'(x) -\deg g_{12}'(x)}g_{12}'^*(x),g_{11}'^*(x))\rangle_e\\
     &=\langle g_{11}(x), -\lambda x^{m+\deg g_{11}'(x)-\deg g_{12}'(x)}g_{12}'^*(x)\rangle_e+\langle g_{12}(x),g_{11}'^*(x)\rangle_e\\
     &=\langle g(x)g_{11}'(x), -\lambda x^{m+\deg g_{11}'(x)-\deg g_{12}'(x)}g_{12}'^*(x)\rangle_e+\langle g(x)g_{12}'(x),g_{11}'^*(x)\rangle_e \\ 
     &=\langle g(x),-\lambda  x^{m+\deg g_{11}'(x)-\deg g_{12}'(x)} g_{11}'^{\circ}(x)g_{12}'^*(x)\rangle_e+\langle g(x), g_{12}'^\circ(x)g_{11}'^*(x)\rangle_e 
     \end{split}
     \end{equation*}
by Lemma \ref{hat} and since $g_{11}(x)=g(x)g_{11}'(x)$. Further, by Lemma \ref{properties}, we know that    $g_{11}'^\circ(x)=\lambda x^{m-\deg g_{11}'(x)} g_{11}^*(x)$ and $g_{12}'^\circ(x)=\lambda x^{m-\deg g_{12}'(x)} g_{12}^*(x)$,    
     $$=\langle g(x), -\lambda^2 x^{m+\deg g_{11}'(x)-\deg g_{12}'(x)} x^{m-\deg g_{11}'(x)} g_{11}'^*(x)g_{12}'^*(x)\rangle_e+\langle g(x),  \lambda x^{m-\deg g_{12}'(x)} g_{12}'^*g_{11}'^*(x)\rangle_e$$
          $$=\langle g(x), -\lambda^2 x^{2m-\deg g_{12}'(x)} g_{11}'^*(x)g_{12}'^*(x)\rangle_e+\langle g(x), \lambda x^{m-\deg g_{12}'(x)} g_{12}'^*(x)g_{11}'^*(x)\rangle_e$$
     $$=\langle g(x), - x^{-\deg g_{12}'(x)} g_{11}'^*(x)g_{12}'^*(x)+x^{-\deg g_{12}'(x)} g_{12}'^*(x)g_{11}'^*(x)\rangle_e=\langle g(x),0\rangle=0,$$
since  $x^m=\lambda^{-1}$ in  $R^*$.     
Thus, $D\subseteq C^{\perp_e}$. This completes the proof.
\end{proof}

\begin{remark}
In \cite{Abdukhalikov2023}, the authors described the dual of one-generator quasi-cyclic codes. 
\end{remark}

\begin{example}\label{exame}
Let $\mathbb{F}_q=\mathbb{F}_5,\ \lambda=2,\  m=11$. Then $\lambda^{-1}=3$ and $$x^{m}-2=(x + 2)(x^5 + x^4 + x^3 + 2x^2 + x + 2)(x^5 + 2x^4 + x^3 + 2x^2 + 3x + 2)\in \mathbb{F}_5[x].$$  
Let $C\subseteq R^2$ be a $2$-quasi-twisted code of length $2m=22$ and index $2$, generated by elements $g_1=(g_{11}(x),g_{12}(x))$ and $g_2=(0,g_{22}(x))$, where
\begin{equation*}
    \begin{split}
       g_{11}(x)=&x+2,\\
       g_{12}(x)=& x^6 +4x^5 +2x^3 +3x^2 +x+3 \\
       g_{22}(x)=&(x^5 +x^4 +x^3 +2x^2 +x+2)(x^5 +2x^4 +x^3 +2x^2 +3x+2) 
    \end{split}
\end{equation*}
The parameters of $C$ are $[22,11,8]$. They are same as those of Best Known Linear Codes (BKLC) \cite{codetable}. The Euclidean dual $C^{\perp_e}$ of $C$ is a $3$-quasi-twisted code of length $2m=22$ and index $2$, generated by elements $f_1=(f_{11}(x),0)$ and $f_2=(f_{21}(x),f_{22}(x))$, where 
    $$ f_{11}(x)=\frac{x^{11}-3}{g_{11}^*(x)}=3x^{10} + x^9 + 2x^8 + 4x^7 + 3x^6 + x^5 + 2x^4 + 4x^3 + 3x^2 + x + 2$$
\begin{equation*}
\begin{split}
      f_{21}(x)&=\frac{-\lambda x^{m+\deg g_{11}(x)-\deg g_{12}(x)}(x^m-\lambda^{-1})g^*_{12}(x)}{g^*_{11}(x)g^*_{22}(x)}\\
      &=3x^{12} + x^{11} + 3x^{10} + 2x^9 + 4x^7 + x^6\\
&\equiv x^9 + 3x^8 + 3x^6 + 4x^5 + 3x^4 + x^3 + 2x^2 + 3x + 1 \pmod{f_{11}(x)}\\
f_{22}(x)= \frac{x^{11}-3}{g^*_{22}(x)}=&4x+2.
\end{split}
\end{equation*}
 The parameters of $C^{\perp_e}$ are $[22,11,8]$.
\end{example}

\section{Symplectic dual of quasi-twisted codes}\label{dualcodess}
Let $C$ be a $\lambda$-quasi-twisted code of length $2m$ and index $2$. Let $a(x),a'(x),b(x),b'(x)$ be polynomials of degree at most $m-1$. Then we define the symplectic inner product as follows
\begin{equation}\label{symdef}
    \langle (a(x),b(x)),(a'(x),b'(x))\rangle_s=\langle a(x),b'(x)\rangle_e-\langle b(x),a'(x)\rangle_e.
\end{equation}
If $C$ is a code in $R^2$, then its symplectic dual is defined as 
$$C^{\perp_s}=\{d\in (R^*)^2 | \ \langle c, d \rangle_s=0 \ \forall c\in C\}.$$
From \cite{ezerman2024}, if $C$ 
is a $\lambda$-quasi-twisted code of index $2$ then $C^{\perp_s}$ is a $\lambda^{-1}$-quasi-twisted codes of index $2$. Now, we provide a polynomial characterization of the symplectic dual of quasi-twisted codes with index 2.

\begin{theorem}\label{thmsdual}
   Let $C$ be a $\lambda$-quasi-twisted code of length $2m$ and index $2$, generated by $g_1=(g_{11}(x),g_{12}(x))$ and $g_2=(0,g_{22}(x))$, satisfying Eq. \ref{eqqt}. Then the symplectic dual $C^{\perp_s}$ is a $\lambda^{-1}$-quasi-twisted code of length $2m$ and index $2$ generated by   
   $\left ( 0,\frac{x^m-\lambda^{-1}}{g_{11}^*(x)}\right)$ and 
   $\left ( \frac{x^m-\lambda^{-1}}{g_{22}^*(x)}, \lambda x^{m-\deg g_{12}+\deg g_{11}}\frac{(x^m-\lambda^{-1})g_{12}^*(x)}{g_{11}^*(x)g_{22}^*(x)}  \right )$.
\end{theorem}
\begin{proof}
    It is clear that $C^{\perp_s}$ is a $\lambda^{-1}$-quasi-twisted code. Suppose $C'\subset  R^*{^2}$ is a $\lambda^{-1}$-quasi-twisted code generated by $f_1=(f_{11}(x),f_{12}(x))$ and $f_2=(0,f_{22}(x))$, where
    $$f_{11}(x)= \frac{x^m-\lambda^{-1}}{g^*_{22}(x)}, f_{12}(x)=\lambda x^{m-\deg g_{12}+\deg g_{11}}\frac{(x^m-\lambda^{-1})g_{12}^*(x)}{g_{11}^*(x)g_{22}^*(x)}, \text{ and }f_{22}(x)=\frac{x^m-\lambda^{-1}}{g_{11}^*(x)} .$$
We show that $C'=C^{\perp_s}$. Similar to Theorem \ref{thmedual}, we have

\begin{equation*}
\begin{split}
   \langle  g_1,f_1\rangle_s&=\langle (g_{11}(x),g_{12}(x)),(f_{11}(x),f_{12}(x))\rangle_s\\
            &=\langle g_{11}(x),f_{12}(x)\rangle_e-\langle g_{12}(x),f_{11}(x)\rangle_e \\
            &=\langle 1,f_{12}(x)g_{11}^{\circ}(x)-f_{11}(x)g_{12}^{\circ}(x)\rangle_e=0, \\
        \langle  g_1,f_2\rangle_s&=0,\\
        \langle  g_2,f_1\rangle_s&=0,\\
        \langle  g_2,f_2\rangle_s&=0.
\end{split}
\end{equation*}
Thus $C'\subseteq C^{\perp_s}$. Also,
since  $\deg g_{11}(x)=\deg g^*_{11}(x)$  and  $\deg g_{22}(x)=\deg g^*_{22}(x))$ therefore
    \begin{equation*}
        \begin{split}
            \dim(C)+\dim(C')&=2m-\deg g_{11}(x)-\deg g_{22}(x)+2m-\deg f_{11}(x)-\deg f_{22}(x)\\
            &=2m-\deg g_{11}(x)-\deg g_{22}(x)\\
            &+2m- (m-\deg g^*_{11}(x))- (m-\deg g^*_{22}(x))\\
            &=2m\  
        \end{split}
    \end{equation*}
Thus, $C'=C^{\perp_s}$. This completes the proof.
\end{proof}

Now, we describe the one-generator case.

\begin{theorem}\label{gen1sdaul}
Suppose $\gcd(m,q)=1$. Let $C$ be one-generator $\lambda$-quasi-twisted code of length $2m$ and index $2$, generated by $(g_{11}(x),g_{12}(x))$ with $g_{11}(x)\mid x^m-\lambda$. Let $g(x)=\gcd(g_{11}(x),g_{12}(x))$, $g_{11}(x)=g(x)g_{11}'(x)$ and $g_{12}(x)=g(x)g_{12}'(x)$. Then the symplectic dual of $C$ is a two-generator $\lambda^{-1}$-quasi-twisted code of length $2m$ and index $2$, generated by 
 $$  (g_{11}'^*(x), \lambda x^{m+\deg g_{11}'(x)-\deg g_{12}'(x)}g_{12}'^*(x)) \text{ and }\left (0, \frac{x^m-\lambda^{-1}}{g_{11}^*(x)}\right ).$$
\end{theorem}
\begin{proof}
Let $D$ be a $\lambda^{-1}$-quasi-twisted code of length $2m$ and index $2$, generated by 
 $$(g_{11}'(x)^*,\lambda x^{m+\deg g_{11}'(x)-\deg g_{12}'(x)}g_{12}'^*(x)) \text{ and } \left (0,\frac{x^m-\lambda^{-1}}{g_{11}^*(x)}\right )  .$$ Since $\gcd((x^m-\lambda^{-1})/g_{11}^*(x), g_{11}'^*(x))=1 $, therefore, by Theorem \ref{qtcode}, 
$$\dim (D)=2m-(m-\deg g_{11}(x))-\deg g_{11}'(x)=m+\deg g(x).$$
Futhermore, $\dim(C^{\perp_s})=2m-\dim(C)=2m-(m-\deg g(x))=m+\deg g(x)$, that is, $\dim (C^{\perp_s})=\dim(D).$ 
Also, 
\begin{align*}
    &\left \langle (g_{11}(x),g_{12}(x)), (g_{11}'^*(x), \lambda x^{m+\deg g_{11}'(x)-\deg g_{12}'(x)}g_{12}'^*(x))  \right \rangle_s\\
    &=\left \langle g_{11}(x), \lambda x^{m+\deg g_{11}'(x)-\deg g_{12}'(x)}g_{12}'^*(x) \right \rangle_e-\langle g_{11}'^*(x), g_{12}(x)\rangle_e\\
    &=0 \ \ \ \ \  \text{ (from the proof of Theorem \ref{gen1edaul})}
\end{align*} 
and  $$\left \langle (g_{11}(x),g_{12}(x)), \left (0,\frac{x^m-\lambda^{-1}}{g_{11}^*(x)}\right ) \right \rangle_s=\left \langle g_{11}(x), \frac{x^m-\lambda^{-1}}{g_{11}^*(x)}   \right \rangle_e=0.$$
Thus, $D\subseteq C^{\perp_s}$. This completes the proof.
\end{proof}

\begin{example}
    Consider the code $C$ as in Example \ref{exame}. Then the symplectic dual $C^{\perp_s}$ of $C$ is a $3$-quasi-twisted code of length $2m=22$ and index $2$, generated by elements $f_1=(f_{11}(x),f_{12}(x))$ and $f_2=(0,f_{22}(x))$, where
    \begin{equation*}
        \begin{split}
            f_{11}(x)=&4x+2,\\
            f_{22}(x)=&3x^{10} + x^9 + 2x^8 +4x^7 +3x^6 +x^5 +2x^4 +4x^3 +3x^2 +x+2,\\
            f_{12}(x)=&2x^{12} + 4x^{11} + 2x^{10} + 3x^9 + x^7 + 4x^6\\
            \equiv&  4x^9 + 2x^8 + 2x^6 + x^5 + 2x^4 + 4x^3 + 3x^2 + 2x + 4\pmod{f_{22}(x)}.
        \end{split}
    \end{equation*}
\end{example}

\section{Hermitian dual of quasi-twisted codes}\label{dualcodesh}

In quantum stabilizer code construction, the field conjugation map $\alpha \mapsto \alpha^q$, which generates the Galois group $\text{Gal}(\mathbb{F}_{q^2}/\mathbb{F}_q)$, serves as a tool that bridges commutativity of Pauli stabilizer generators with classical Hermitian inner product over $\mathbb{F}_{q^2}$. For a given basis $\{\beta_1, \beta_2\}$ of $\mathbb{F}_{q^2}$ over $\mathbb{F}_q$, the conjugation map defines the standard Hermitian inner product $\langle \mathbf{u}, \mathbf{v} \rangle_h = \sum_{i=1}^n u_i v_i^q$, which is subsequently used to construct the trace-alternating symplectic form $\langle \mathbf{u}, \mathbf{v}\rangle_a = \mathrm{Tr}_{\mathbb{F}_q/\mathbb{F}_p}( (\langle \mathbf{u},\mathbf{v}\rangle_h - \langle \mathbf{v},\mathbf{u}\rangle_h)/(\beta_1 \beta_2^q - \beta_1^q\beta_2) )$. Because the skew-symmetric difference is scaled by a non-zero element whose conjugate is its negation, the resulting form scales the complex-like Hermitian structure into the base field, ensuring that the alternating dual of a linear code $C$ exactly coincides with its Hermitian dual ($C^{\perp_a} = C^{\perp_h}$). By establishing this equivalence, the conjugation map dictates that the symplectic orthogonality required for Pauli stabilizer generators is satisfied precisely when a classical $[n, k]$ code $C$ over $\mathbb{F}_{q^2}$ is Hermitian self-orthogonal ($C \subseteq C^{\perp_h}$), thereby yielding an $[[n, n-2k]]$ quantum stabilizer code over $\mathbb{F}_q$. For details see \cite{ezerman2024,ketkar2006}. In this section, we consider the field to be $\mathbb{F}_{q^2}$. Analogous to the previous sections, we first provide some properties of the Hermitian inner product in terms of polynomial representations. These properties are subsequently employed to determine the Hermitian duals of quasi-twisted codes.

Denote $\overline{R}^{*}=\mathbb{F}_{q^2}[x]/\langle (x^m-\lambda^{-q})\rangle$.
For $a(x)=a_0+a_1x+\dots +a_{m-1}x^{m-1}, b(x)=b_0+b_1x+\dots +b_{m-1}x^{m-1}\in \mathbb{F}_{q^2}[x]$ and $\alpha\in \mathbb{F}_{q^2}$. Define the conjugate of $\alpha$ as $\overline{\alpha}=\alpha^q$ and
$$\overline{a}(x)=\overline{a}_0+\overline{a}_1x+\dots +\overline{a_{m-1}}x^{m-1}.$$
The Hermitian inner product is defined as
$$\langle a(x),b(x)\rangle_h=\langle a(x),\overline{b}(x)\rangle_e=a_0\overline{b_0}+a_1\overline{b_1}+\dots+a_{m-1}\overline{b_{m-1}}.$$
To perform the Hermitian inner product for polynomials of degree greater than or equal to $m$, we consider the polynomials in $R$ or $\overline{R}^*$. It will be clarified from the context.  This inner product naturally extends componentwise over $R^l$ and $\overline{R}^*{^l}$.
Note that the pairing $\langle \cdot, \cdot\rangle_h$ is nondegenerate, that is, \begin{align*}
    \langle a(x), b(x)\rangle_h=0 \ \forall a(x)\in R \implies b(x)=0 \in R^*.
\end{align*}

\begin{definition}
    Let $t(x)=t_0+t_1x+\dots+t_dx^d\in \mathbb{F}_{q}[x]$ with $t_d\neq 0$ be a nonzero polynomial, then the conjugate reciprocal of $t(x)$ is denoted by $\overline{t}^*(x)$ and is defined as $\overline{t}^*(x)=x^d\overline{t}(1/x)$ that is $\overline{t}^*(x)=\overline{t^*(x)}$. We define $\overline{t^*(x)}=0$ whenever $t(x)=0$.
\end{definition}

Note that $\overline{(x^m-\lambda)}^*=(1-\lambda^q x^m)=-\lambda^{q} (x^m-\lambda^{-q})$. It is easy to see that $\overline{t^*}(x)$ divides $(x^m-\lambda^{-q})$ whenever $t(x)$ divides $(x^m-\lambda)$. From Lemma \ref{properties}, we have the following.
\begin{lemma}\label{h-properties}
\begin{enumerate}[(1)]
    \item  If $t(x)=t_0+t_1x+\dots+t_dx^d$ with $t_d\neq 0$ and $d\leq m$, then $\overline{t^\circ}(x)=\lambda^q x^{m-d} \overline{t}^*(x)$.
 \item If $t(x)=t_0+t_1x+\dots+t_mx^m$ then $$\overline{t^\circ}(x)=\lambda^q x^m\overline{t}(1/x)\equiv (\overline{t}_0+\lambda^q\overline{ t}_m)+ \lambda^q (\overline{t}_{m-1}x+\dots \overline{t}_2x^{m-2}+\overline{t}_1x^{m-1})\pmod {x^m-\lambda^{-q}}.$$
   \end{enumerate}
\end{lemma}
\begin{proof}
(1) $\overline{t^{\circ}}(x)=\overline{\lambda x^m t(1/x)}=\lambda^q x^{m-d}x^d\overline{t}(1/x)=\lambda^q x^{m-d}\overline{t}^*(x)$.\\
 (2) We have \begin{equation*}
 \begin{split}
   \overline{t^\circ}(x)&=\overline{\lambda x^mt(1/x)}=\lambda^q (t_0^qx^m+t_1^qx^{m-1}+\dots+t_{m-2}^qx^2+t_{m-1}^qx+t_m^q)\\
 &\equiv (t_0^q+\lambda^q t_m^q)+ \lambda^q (t_{m-1}^qx+\dots t_2^qx^{m-2}+t_1^qx^{m-1})\pmod {x^m-\lambda^{-q}},
 \end{split}
 \end{equation*}
which completes the proof.
\end{proof}
\begin{lemma}\label{h-innerchange}
    Let $a(x),b(x),t(x)\in \mathbb{F}_{q^2}[x]$. Then
    $$\langle a(x)t(x),b(x)\rangle_h=\langle a(x),b(x)\overline{t^\circ}(x)\rangle_h,$$
    where $a(x),a(x)t(x)$ and $b(x),b(x)\overline{t^\circ}(x)$ are considered in $R$ and $\overline{R}^*$, respectively.
\end{lemma}
\begin{proof}
We have
    \begin{equation*}
      \langle a(x)t(x),b(x)\rangle_h=\langle a(x)t(x),\overline{b}(x)\rangle_e=\langle a(x),\overline{b}(x)t^{\circ}(x)\rangle_e=\langle a(x),b(x)\overline{t^{\circ}}(x)\rangle_h, 
    \end{equation*}
which completes the proof. 
\end{proof}

If $C$ is a code in $R^l$, then its Hermitian dual is defined as 
$$C^{\perp_h}=\{d\in (\overline{R}^*)^l | \ \langle c, d \rangle_h=0 \ \forall c\in C\}.$$

\begin{proposition}\cite[Proposition 6.2]{sangwisut2017}\label{hqdual}
    If $C$ is a $\lambda$-QT code of index $l$ then $C^{\perp_h}$ is a $\lambda^{-q}$-QT code of index $l$. 
\end{proposition}

Now, we explicitly construct the tuples of polynomials that generate the Hermitian duals of quasi-twisted codes with index $2$.
\begin{theorem}\label{thmhdual}
    Let $C$ be a $\lambda$-quasi-twisted code of length $2m$ and index $2$ over $\mathbb{F}_{q^2}$, generated by $g_1=(g_{11}(x),g_{12}(x))$ and $g_2=(0,g_{22}(x))$, satisfying Eq. \ref{eqqt}. Then the Hermitian dual $C^{\perp_h}$ is a $\lambda^{-q}$-quasi-twisted code of length $2m$ and index $2$ generated by $\left (\frac{x^m-\lambda^{-q}}{\overline{g_{11}}^*(x)},0\right )$ and $\left (-\lambda^q x^{m-\deg g_{12}+\deg g_{11}}\frac{(x^m-\lambda^{-q})\overline{g}^*_{12}(x)}{\overline{g_{11}}^*(x)\overline{g_{22}}^*(x)},\frac{x^m-\lambda^{-q}}{\overline{g_{22}}^*(x)}\right )$.
\end{theorem}
\begin{proof}
  By Proposition \ref{hqdual}, $C^{\perp_h}$ is a $\lambda^{-q}$-quasi-twisted code of length $2m$ and index $2$. Assume $C'\subseteq \overline{R}^*$ is $\lambda^{-q}$-quasi-twisted code of length $2m$ and index $2$, generated by elements $f_1=(f_{11}(x),0)$ and $f_2=(f_{21}(x),f_{22}(x))$, where
  $$f_{11}(x)=\frac{x^m-\lambda^{-q}}{\overline{g_{11}}^*(x)}, 
  \ f_{22}(x)=\frac{x^m-\lambda^{-q}}{\overline{g_{22}}^*(x)},$$
  $$f_{21}(x)=-\lambda^q x^{m-\deg g_{12}+\deg g_{11}}\frac{(x^m-\lambda^{-q})\overline{g}^*_{12}(x)}{\overline{g_{11}}^*(x)\overline{g_{22}}^*(x)}.$$ 
We show that $C'=C^{\perp_h}$.  Note that
 \begin{equation*}
 \begin{split}
      \dim (C')&=2m-\deg f_{11}(x)-\deg f_{22}(x)\\
      &=2m-(m-\deg \overline{g_{11}}^*(x))-(m-\deg \overline{g_{22}}^*(x))\\
      &= 2m-(m-\deg g_{11}(x))-(m-\deg g_{22}(x))\\
      &=\deg g_{11}(x)+\deg g_{22}(x)\\
      &=2m-\dim (C).
       \end{split}
 \end{equation*}
Also, using Lemma \ref{h-properties}, we have
\begin{equation*}
    \begin{split}
        &f_{11}(x) \overline{g_{11}^{\circ}}(x)=\frac{x^m-\lambda^{-q}}{\overline{g}^*_{11}(x)} \cdot \lambda^q x^{m-\deg g_{11}(x)}\overline{g}^*_{11}(x)=\lambda^q x^{m-\deg g_{11}(x)}(x^m-\lambda^{-q}) =0 \text{ in } \overline{R}^*\\
        &f_{22}(x) \overline{g_{22}^{\circ}}(x)=\frac{x^m-\lambda^{-q}}{\overline{g}^*_{22}(x)} \cdot \lambda^q x^{m-\deg g_{22}(x)}\overline{g}^*_{22}(x)=\lambda^q x^{m-\deg g_{22}(x)}(x^m-\lambda^{-q}) =0 \text{ in } \overline{R}^*
           \end{split}
\end{equation*} 
        $f_{21}(x)\overline{g^{\circ}_{11}}(x)+f_{22}(x)\overline{g_{12}^{\circ}}(x)=$        
        $$\frac{-\lambda^q x^{m+\deg g_{11}(x)-\deg g_{12}(x)}(x^m-\lambda^{-q})\overline{g}^*_{12}(x)}{\overline{g}^*_{11}(x)\overline{g}^*_{22}(x)}\lambda^q x^{m-\deg g_{11}(x)}\overline{g}^*_{11}(x)
        + \frac{x^m-\lambda^{-q}}{\overline{g}^*_{22}(x)} \lambda^q x^{m-\deg g_{12}(x)} \overline{g}^*_{12}(x)$$
        $$=\frac{-\lambda^{2q} x^{2m-\deg g_{12}(x)}(x^m-\lambda^{-q})\overline{g}^*_{12}(x)}{\overline{g}^*_{22}(x)}
        + \frac{x^m-\lambda^{-q}}{\overline{g}^*_{22}(x)} \lambda^q x^{m-\deg g_{12}(x)} \overline{g}^*_{12}(x)$$
     since $x^m=\lambda^{-q}$  in  $\overline{R}^*$, 
        $$=\frac{-\lambda^{2q} x^{2m-\deg g_{12}(x)}(x^m-\lambda^{-q})\overline{g}^*_{12}(x)}{\overline{g}^*_{22}(x)}
        + \frac{(x^m-\lambda^{-q})}{\overline{g}^*_{22}(x)} \lambda^q x^{m-\deg g_{12}(x)} \overline{g}^*_{12}(x)=0 \text{ in } \overline{R}^*$$    
 
since $x^m=\lambda^{-q}$ in $ \overline{R}^*$. 
Thus, we have  $\langle g_2, f_1\rangle_h=0$ (trivially)  and
\begin{equation*}
\begin{split}
        \langle g_1, f_1\rangle_h&= \langle g_{11}(x), f_{11}(x)\rangle_h= \langle 1, f_{11}(x) \overline{g_{11}^{\circ}}(x)\rangle_h=0,\\
        \langle g_2, f_2\rangle_h&= \langle g_{22}(x), f_{22}(x)\rangle_h= \langle 1, f_{22}(x) \overline{g_{22}^{\circ}}(x)\rangle_h=0,\\
       \langle g_1,f_2\rangle_h&= \langle  g_{11}(x),f_{21}(x)\rangle_h+\langle  g_{12}(x),f_{22}(x)\rangle_h\\
       &=\langle 1, f_{21}(x)\overline{g^{\circ}_{11}}(x)+f_{22}(x)\overline{g_{12}^{\circ}}(x)\rangle_h=0.
        \end{split}
    \end{equation*}

    Hence, $C'\subseteq C^{\perp_h}$. This completes the proof.
 \end{proof} 
The next result gives the Hermitian dual of one-generator codes.

    \begin{theorem}\label{gen1hdaul}
Suppose $\gcd(m,q)=1$. Let $C$ be one-generator $\lambda$-quasi-twisted code of length $2m$ and index $2$, generated by $(g_{11}(x),g_{12}(x))$ with $g_{11}(x)\mid x^m-\lambda$. Let $g(x)=\gcd(g_{11}(x),g_{12}(x))$, $g_{11}(x)=g(x)g_{11}'(x)$ and $g_{12}(x)=g(x)g_{12}'(x)$. Then Hermitian dual of $C$ is a two-generator $\lambda^{-q}$-quasi-twisted code of length $2m$ and index $2$, generated by 
 $$\left (\frac{x^m-\lambda^{-q}}{\overline{g_{11}}^*(x)},0\right ) \text{ and } (-\lambda^q x^{m+\deg g_{11}'(x)-\deg g_{12}'(x)}\overline{g_{12}'}^*(x),\overline{g_{11}'}^*(x)).$$
\end{theorem}
\begin{proof}
 Let $D$ be $\lambda^{-q}$-quasi-twisted code generated by 
 $$\left (\frac{x^m-\lambda^{-q}}{\overline{g_{11}}^*(x)},0\right ) \text{ and } (-\lambda^q x^{m+\deg g_{11}'(x)-\deg g_{12}'(x)}\overline{g_{12}'}^*(x),\overline{g_{11}'}^*(x)).$$ 
 Using similar arguments as in the proof of Theorem \ref{gen1edaul}, we have $\dim (D)=\dim(C^{\perp_h})$. Furthermore, 
 \begin{align*}
     \left \langle  (g_{11}(x),g_{12}(x)), \left (\frac{x^m-\lambda^{-q}}{\overline{g_{11}}^*(x)},0 \right )  \right \rangle_h=&\left \langle g_{11}(x), \overline{\left ( \frac{x^m-\lambda^{-q}}{\overline{g_{11}}^*(x)} \right )}   \right \rangle_e\\
     =&\left \langle g_{11}(x),  \frac{x^m-\lambda^{-1}}{g_{11}^*(x)}    \right \rangle_e=0
 \end{align*}
 and 
 \begin{align*}
     &\left \langle (g_{11}(x),g_{12}(x)), (-\lambda^q x^{m+\deg g_{11}'(x)-\deg g_{12}'(x)}\overline{g_{12}'}^*(x),\overline{g_{11}'}^*(x))  \right \rangle_h\\
     =&\left \langle (g_{11}(x),g_{12}(x)), (-\lambda x^{m+\deg g_{11}'(x)-\deg g_{12}'(x)}g_{12}'^*(x),g_{11}'^*(x))  \right \rangle_e\\
     =& 0 \ \ \  \text{( from the proof of Theorem \ref{gen1edaul}).}
 \end{align*}
Thus, $D\subseteq C^{\perp_h}$. This completes the proof.   
\end{proof}

\begin{example}
Let $q=2$, $\mathbb{F}_{q^2}=\mathbb{F}_2[w],$ with $w^2+w+1=0$, $\ \lambda=w,\  m=11$. Then $\lambda^{-q}=w$ and $$x^{m}-w=(x + w^2)(x^5 + x^4 + wx^3 + x^2 + wx + w)(x^5 + wx^4 + wx^3 + x^2 + x + w)\in \mathbb{F}_4[x].$$  
Let $C\subseteq R^2$ be a $\lambda$-quasi-twisted code of length $2m=22$ and index $2$, generated by elements $g_1=(g_{11}(x),g_{12}(x))$ and $g_2=(0,g_{22}(x))$, where
\begin{equation*}
    \begin{split}
       g_{11}(x)=&x+w^2,\\
       g_{12}(x)=&x^2+x+w^2, \\
       g_{22}(x)=&x^5 + wx^4 + wx^3 + x^2 + x + w. 
    \end{split}
\end{equation*}
The parameters of $C$ are $[22,16,4]$ (BKLC). The Hermitian dual $C^{\perp_h}$ of $C$ is a $\lambda^{-q}$-quasi-twisted code of length $2m=22$ and index $2$, generated by elements $f_1=(f_{11}(x),0)$ and $f_2=(f_{21}(x),f_{22}(x))$, where 
\begin{equation*}
\begin{split}
     f_{11}(x)=\frac{x^{11}-w}{\overline{g_{11}}^*(x)}=& w^2x^{10} + wx^9 + x^8 + w^2x^7 + wx^6 + x^5 + w^2x^4 + wx^3 + x^2 + w^2x +w,\\
      f_{21}(x)=&\frac{-\lambda^q x^{m-\deg g_{12}(x)+\deg g_{11}(x)}(x^m-\lambda^{-q})\overline{g}^*_{12}(x)}{\overline{g}^*_{11}(x)\overline{g}^*_{22}(x)}\\\\
      =&x^{17} + wx^{16} + wx^{15} + w^2x^{14} + x^{10}\\
\equiv& w^2x^9 + wx^8 + x^7 + x^6 + x^5 + wx^4 + wx^3 + wx^2 + x + w^2 \pmod{f_{11}(x)},\\
f_{22}(x)= \frac{x^{11}-w}{\overline{g}^*_{22}(x)}=&wx^6 + w^2x^5 + wx^4 + wx^2 + x + w.
\end{split}
\end{equation*}
 The parameters of $C^{\perp_h}$ are $[22,6,11]$.
\end{example}

 \section{Self-orthogonality of quasi-twisted codes}\label{selfdualcodes}
By Proposition \ref{ldual} and  Theorem \ref{thmsdual}, $C^{\perp_e}$ and $C^{\perp_s}$ are $\lambda^{-1}$-quasi-twisted codes whenever $C$ is a $\lambda$-quasi-twisted code. Thus, we discuss the self-orthogonality of $C$ with respect to the Euclidean and symplectic inner products only when $\lambda=\lambda^{-1}$, that is, $\lambda=\pm 1$. When $\lambda=1$, the code is called a quasi-cyclic, and when $\lambda=-1$, the code is called a quasi-negacyclic code.
\begin{theorem}\label{selforthogonaleuclidean}
Let $\lambda=\pm 1$.  Let $C$ be a $\lambda$-quasi-twisted code of length $2m$ and index $2$ generated by elements $g_1=(g_{11}(x),g_{12}(x))$, $g_2=(0,g_{22}(x))$ satisfying Eq \ref{eqqt}. Then $C$ is self-orthogonal with respect to the Euclidean inner product, that is, $C\subseteq C^{\perp_e}$ if and only if \\
1)    $g_{22}(x)g_{22}^*(x)\equiv 0\pmod{x^m-\lambda}$,\\
 2)     $g_{22}(x)g_{12}^*(x)\equiv 0\pmod{x^m-\lambda}$,\\
 3)  $x^{\deg g_{12}(x)} g_{11}(x)g_{11}^*(x)+x^{\deg g_{11}(x)}g_{12}(x)g_{12}^*(x)\equiv 0\pmod{x^m-\lambda}$.
\end{theorem}
\begin{proof}
    $C$ is self-orthogonal, i.e. $C\subseteq C^{\perp_e}$ if and only if
\begin{equation*}
    \begin{split}
        \langle a(x)g_1,g_1\rangle_e&=\langle (a(x)g_{11}(x),a(x)g_{12}(x)),(g_{11}(x),g_{12}(x))\rangle_e=0,\\
        \langle a(x)g_1,g_2\rangle_e&=\langle (a(x)g_{11}(x),a(x)g_{12}(x)),(0,g_{22}(x)) \rangle_e=0,\\
        \langle a(x)g_2,g_2\rangle_e&=\langle (0,a(x)g_{22}(x)),(0,g_{22}(x)) \rangle_e=0
    \end{split}
\end{equation*}
for all $a(x)\in R$. Equivalently,
\begin{equation*}
    \begin{split}
   \langle a(x),g_{11}(x)g_{11}^\circ(x)+g_{12}(x)g_{12}^\circ(x)\rangle_e=0,\\
        \langle a(x), g_{22}(x)g_{12}^\circ(x)\rangle_e=0,\\
        \langle a(x)),g_{22}(x)g_{22}^\circ(x) \rangle_e=0
    \end{split}
\end{equation*}
for all $a(x)\in R$. Equivalently,
\begin{equation*}
    \begin{split}
        g_{11}(x)g_{11}^\circ(x)+g_{12}(x)g_{12}^\circ(x)\equiv 0\pmod{x^m-\lambda},\\
        g_{22}(x)g_{12}^\circ(x)\equiv 0\pmod{x^m-\lambda},\\
        g_{22}(x)g_{22}^\circ(x)\equiv 0\pmod{x^m-\lambda}.
    \end{split}
\end{equation*}
Using the fact $t^\circ(x)= \lambda x^{m-\deg t(x)}t^*(x)$, equivalently, we have
\begin{equation*}
    \begin{split}
        x^{\deg g_{12}(x)} g_{11}(x)g_{11}^*(x)+x^{\deg g_{11}(x)}g_{12}(x)g_{12}^*(x)\equiv 0\pmod{x^m-\lambda},\\
        g_{22}(x)g_{12}^*(x)\equiv 0\pmod{x^m-\lambda},\\
        g_{22}(x)g_{22}^*(x)\equiv 0\pmod{x^m-\lambda}.
    \end{split}
\end{equation*}
This completes the proof.
\end{proof}

\begin{example}
Let $q=3$, $\lambda=-1$, $m=13$.   Let $C$ be $\lambda$-quasi-twisted code of length $2m$ and index $2$, generated by elements $g_1=(g_{11}(x),g_{12}(x))$ and $g_2=(0,g_{22}(x))$, where
\begin{equation*}
    \begin{split}
       g_{11}(x)=&(x+1)(x^3+2x+1),\\
       g_{12}(x)=&(x+1)(x^3+2x+1)(2x^3 + x^2 + x + 1),\\
       g_{22}(x)=&(x+1)(x^3+2x+1)(x^3+x^2+2x+1)(x^3+2x^2+x+1). 
    \end{split}
\end{equation*}
The parameters of $C$ are $[26,12,9]$ (BKLC).
Observe that  
\begin{align*}
     g_{22}(x)g_{22}^*(x)=&x^{20} + x^{19} + 2x^{18} + 2x^{15} + x^{14} + x^{13} + x^7 + x^6 + 2x^5 + 2x^2
    + x + 1\\
    \equiv& 0\pmod{x^m-\lambda},\\
 g_{22}(x)g_{12}^*(x)=&x^{17} + 2x^{16} + 2x^{15} + 2x^{13} + x^4 + 2x^3 + 2x^2 + 2\\
 \equiv& 0\pmod{x^m-\lambda},
\end{align*}
and 
$$x^{\deg g_{12}(x)} g_{11}(x)g_{11}^*(x)+x^{\deg g_{11}(x)}g_{12}(x)g_{12}^*(x)=2x^{18} + 2x^{17} + 2x^5 + 2x^4\\\equiv 0\pmod{x^m-\lambda}.$$
Thus, by Theorem \ref{selforthogonaleuclidean}, $C$ is Euclidean self-orthogonal.
\end{example}

\begin{theorem}\label{selfsym}
 Let $\lambda=\pm 1$. Let $C$ be a $\lambda$-quasi-twisted code of length $2m$ and index $2$ generated by elements $g_1=(g_{11}(x),g_{12}(x))$, $g_2=(0,g_{22}(x))$ satisfying Eq \ref{eqqt}. Then $C$ is self-orthogonal with respect to the symplectic inner product that is $C\subseteq C^{\perp_s}$ if and only if \\
1) $g_{11}(x)g_{22}^*(x)\equiv 0\pmod{x^m-\lambda}$,\\
   2)  $x^{\deg g_{11}(x)} g_{11}(x)g_{12}^*(x)-x^{\deg g_{12}(x)}g_{12}(x)g_{11}^*(x)\equiv 0\pmod{x^m-\lambda}$.
\end{theorem}
\begin{proof}
$C$ is self-orthogonal with respect to the symplectic inner product if and only if
\begin{align*}
    \langle a(x) g_1,g_1\rangle_s=&\langle a(x)g_{11}(x),g_{12}(x)\rangle_e-\langle a(x)g_{12}(x),g_{11}(x)\rangle_e=0,\\
    \langle a(x) g_1,g_2\rangle_s=&\langle a(x)g_{11}(x),g_{22}(x)\rangle_e=0\\
    \langle a(x) g_2,g_2\rangle_s=&\langle (0, a(x) g_{22}(x)),(0, a(x) g_{22}(x)\rangle_s=0 \ \ \text{(always holds)}
\end{align*} for all $a(x) \in R$.  Equivalently, by Lemma \ref{hat}, 
\begin{align*}
    \langle a(x),g_{11}^\circ(x)g_{12}(x)\rangle_e-\langle a(x),g_{12}^\circ(x)g_{11}(x)\rangle_e=0\\
   \langle a(x), g_{11}^\circ(x)g_{22}(x)\rangle_e=0
\end{align*} for all $a(x)\in R$. Equivalently, 
\begin{align*}
    g_{11}^\circ(x)g_{12}(x)-g_{12}^\circ(x)g_{11}(x)=0 \text{ in } R^*\\
    g_{11}^\circ(x)g_{22}(x)=0 \text{ in } R^*.
\end{align*} We get the desired result using Lemma \ref{properties}. 
\end{proof}

\begin{example}
 Let $q=3$, $\lambda=-1$, $m=10$. Let $C$ be $\lambda$-quasi-twisted code of length $2m$ and index $2$, generated by elements $g_1=(g_{11}(x),g_{12}(x))$ and $g_2=(0,g_{22}(x))$, where
\begin{equation*}
    \begin{split}
       g_{11}(x)=&x^4 + 2x^3 + x + 1,\\
       g_{12}(x)=&x^4 + 2x^3 + x + 1,\\
       g_{22}(x)=&(x^2+1)(x^4 + 2x^3 + x + 1). 
    \end{split}
\end{equation*}
The parameters of $C$ are $[20,10,6]$.
Observe that 
\begin{align*}
   g_{11}(x)g_{22}^*(x)&=x^m+1\equiv 0\pmod{x^m-\lambda},\\
   x^{\deg g_{11}(x)} g_{11}(x)g_{12}^*(x)-x^{\deg g_{12}(x)}g_{12}(x)g_{11}^*(x)&=x^4g_{11}(x)g_{12}^*(x)-x^4g_{12}(x)g_{11}^*(x)\\
   &\equiv 0\pmod{x^m-\lambda}. 
\end{align*}
Thus, by Theorem \ref{selfsym}, $C$ is symplectic self-orthogonal. 
\end{example}

By Proposition \ref{hqdual}, $C^{\perp_h}$ is a $\lambda^{-q}$-quasi-twisted code whenever $C$ is a $\lambda$-quasi-twisted code. Thus, we discuss the self-orthogonality of $C$ with respect to the Hermitian inner product only when $\lambda=\lambda^{-q}$, that is $\lambda^{q+1}= 1$.

 \begin{theorem}\label{hermistiansleforthogonal}
   Let $\lambda\in \mathbb{F}_{q^2}$ such that $\lambda^{q+1}=1$. Let $C$ be a $\lambda$-quasi-twisted code of length $2m$ and index $2$ generated by elements $g_1=(g_{11}(x),g_{12}(x))$, $g_2=(0,g_{22}(x))$ satisfying Eq \ref{eqqt}. Then $C$ is self-orthogonal with respect to the Hermitian inner product that is $C\subseteq C^{\perp_h}$ if and only if \\
1) $g_{22}(x)\overline{g_{22}}^*(x)\equiv 0\pmod{x^m-\lambda}$, \\
   2)     $g_{22}(x)\overline{g_{12}}^*(x)\equiv 0\pmod{x^m-\lambda}$,\\
     3) $x^{\deg g_{12}(x)} g_{11}(x)\overline{g_{11}}^*(x)+x^{\deg g_{11}(x)}g_{12}(x)\overline{g_{12}}^*(x)\equiv 0\pmod{x^m-\lambda}$. 
\end{theorem}

\begin{proof}
$C$ is self-orthogonal with respect to the Hermitian inner product if and only if
\begin{equation*}
    \begin{split}
        \langle a(x)g_1,g_1\rangle_h&=\langle (a(x)g_{11}(x),a(x)g_{12}(x)),(g_{11}(x),g_{12}(x))\rangle_h=0,\\
        \langle a(x)g_1,g_2\rangle_h&=\langle (a(x)g_{11}(x),a(x)g_{12}(x)),(0,g_{22}(x)) \rangle_h=0,\\
        \langle a(x)g_2,g_2\rangle_h&=\langle (0,a(x)g_{22}(x)),(0,g_{22}(x)) \rangle_h=0
    \end{split}
\end{equation*}
for all $a(x)\in R$. Equivalently, by Lemma \ref{h-innerchange},
\begin{equation*}
    \begin{split}
   \langle a(x),g_{11}(x)\overline{g_{11}^\circ}(x)+g_{12}(x)\overline{g_{12}^\circ}(x)\rangle_h=0,\\
        \langle a(x), g_{22}(x)\overline{g_{12}^\circ}(x)\rangle_h=0,\\
        \langle a(x)),g_{22}(x)\overline{g_{22}^\circ}(x) \rangle_h=0
    \end{split}
\end{equation*}
for all $a(x)\in R$. Equivalently,
\begin{equation*}
    \begin{split}
        g_{11}(x)\overline{g_{11}^\circ}(x)+g_{12}(x)\overline{g_{12}^\circ}(x)\equiv 0\pmod{x^m-\lambda},\\
        g_{22}(x)\overline{g_{12}^\circ}(x)\equiv 0\pmod{x^m-\lambda},\\
        g_{22}(x)\overline{g_{22}^\circ(x)}\equiv 0\pmod{x^m-\lambda}.
    \end{split}
\end{equation*}
Using Lemma \ref{h-properties}, we get the desired result. 
\end{proof}

\begin{example}
Let $q=2$, $\mathbb{F}_{q^2}=\mathbb{F}_2[w],$ with $w^2+w+1=0$, $ \lambda=w$,  $m=11$.  
Let $C\subseteq R^2$ be a $\lambda$-quasi-twisted code of length $2m=22$ and index $2$, generated by elements $g_1=(g_{11}(x),g_{12}(x))$ and $g_2=(0,g_{22}(x))$, where
\begin{equation*}
    \begin{split}
       g_{11}(x)=&x+w^2,\\
       g_{12}(x)=&(x+w^2)(w^2x^5 + x^4 + x^3 + w^2x^2 + x + 1), \\
       g_{22}(x)=&x^m-w. 
    \end{split}
\end{equation*}
The parameters of $C$ are $[22,10,8]$. It is easy to see that $g_{22}(x)\overline{g_{22}}^*(x)\equiv 0\pmod{x^m-\lambda}$,
   and
$x^{\deg g_{12}(x)} g_{11}(x)\overline{g_{11}}^*(x)+x^{\deg g_{11}(x)}g_{12}(x)\overline{g_{12}}^*(x)\equiv 0\pmod{x^m-\lambda}$.
Thus, by Theorem \ref{hermistiansleforthogonal}, $C$ is Hermitian self-orthogonal. 
\end{example}

A quantum error-correcting code $Q$, with parameters $[[n,k,d]]_q$, is defined as a $K$-dimensional subspace of the $n$-fold tensor product $(\mathbb{C}^q)^{\otimes n}$. Here, $\mathbb{C}^q$ denotes the $q$-dimensional complex Hilbert space, and the dimension of the code satisfies $K = q^k$. The minimum distance $d$ characterizes the error-correcting capability of the code.
For detailed background and constructions of quantum error-correcting codes, see \cite{ezerman2024}.
A well known method to construct quantum stabilizer codes is the CSS construction in which we require two codes $C_1$ and $C_2$ such that $C_1\subseteq C_2^{\perp_e}$ (equivalently, $C_2\subseteq C_1^{\perp_e}$). Next, we discuss the containment of quasi-twisted codes. Consider two codes $C_1\subseteq R^2$ and $C_2\subseteq R^*{^2}$ i.e. $C_1$ is a $\lambda$-quasi-twisted code and $C_2$ is a $\lambda^{-1}$-quasi-twisted code. Then $C_2^{\perp_e}\subseteq R^2$ is a $\lambda$-quasi-twisted code. We discuss the necessary and sufficient conditions for $C_1\subseteq C_2^{\perp_e}$. 

\begin{theorem}\label{c1containc2e}
 Let $C_1\subseteq R^2$ be a $\lambda$-quasi-twisted code of length $2m$ and index $2$ generated by elements $g_1=(g_{11}(x),g_{12}(x))$, $g_2=(0,g_{22}(x))$ satisfying Eq \ref{eqqt} and $C_2\subseteq {R^*}^2$ be a $\lambda^{-1}$-quasi-twisted code of length $2m$ and index $2$ generated by elements $f_1=(f_{11}(x),0)$, $f_2=(f_{21}(x),f_{22}(x))$ such that $f_{11}(x)$ and $f_{22}(x)$ divide $(x^m-\lambda^{-1})$. Then $C_1\subseteq C_2^{\perp_e}\subseteq R^2$ if and only if\\
1) $g^*_{11}(x)f_{11}(x)\equiv 0\pmod{x^m-\lambda^{-1}}$,\\
 2) $g_{22}^*(x)f_{22}(x)\equiv 0\pmod{x^m-\lambda^{-1}}$,\\
 3) $x^{\deg g_{12}(x)}f_{21}(x) g^*_{11}(x)     +x^{\deg g_{11}(x)}f_{22}(x) g^*_{12}(x)\equiv 0\pmod{x^m-\lambda^{-1}}$.
\end{theorem}

\begin{proof}
By Theorem \ref{thmedual}, $C_2^{\perp_e}\subseteq R^2$. 
Note that $C_1\subseteq C_2^{\perp_e}$ means that $\langle c_1,c_2\rangle_e=0$ for all $c_1\in C_1$ and $c_2\in C_2$, equivalently,
\begin{equation*}
    \begin{split}
        \langle a(x)g_1,b(x)f_1\rangle_e&=\langle (a(x)g_{11}(x),a(x)g_{12}(x)),(b(x)f_{11}(x),0)\rangle_e=0,\\\\
        \langle a(x)g_2, b(x)f_2\rangle_e&=\langle (0,a(x)g_{22}(x)),(b(x)f_{21}(x),b(x)f_{22}(x)) \rangle_e=0,\\\\
        \langle a(x)g_1, b(x)f_2\rangle_e&=\langle (a(x)g_{11}(x),a(x)g_{12}(x)),(b(x)f_{21}(x),b(x)f_{22}(x)) \rangle_e=0
    \end{split}
\end{equation*}
for all $a(x)\in R$ and $b(x)\in R^*$. Equivalently,
\begin{equation*}
    \begin{split}
        \langle (a(x), b(x)f_{11}(x)g^{\circ}_{11}(x)) \rangle_e=0,\\\\
        \langle a(x)),b(x)f_{22}(x)g^{\circ}_{22}(x) \rangle_e=0,\\\\
         \langle (a(x),b(x)(f_{21}(x)g^{\circ}_{11}(x)+f_{22}(x)g^{\circ}_{12}(x))\rangle_e=0,
    \end{split}
\end{equation*}
for all $a(x)\in R$ and $b(x)\in R^*$.
Equivalently, 
\begin{equation*}
    \begin{split}
        f_{11}(x)g^{\circ}_{11}(x)=0  \text{ in } R^* \iff  f_{11}(x)g^*_{11}(x)=0 \text{ in } R^* \iff  f_{11}(x)g^*_{11}(x)\equiv 0 \pmod{x^m-\lambda^{-1}}\\
        f_{22}(x)g^{\circ}_{22}(x)=0 \text{ in } R^* \iff f_{22}(x)g^*_{22}(x)=0 \text{ in } R^* \iff  f_{22}(x)g_{22}^*(x)\equiv 0 \pmod{x^m-\lambda^{-1}},
    \end{split}
\end{equation*}
and 
\begin{equation*}
    \begin{split}
       &f_{21}(x)g^{\circ}_{11}(x)+f_{22}(x)g^{\circ}_{12}(x)\equiv 0 \bmod(x^m-\lambda^{-1}) \\
       \\
      \iff & f_{21}(x)\lambda x^{m-\deg g_{11}(x)}g^*_{11}(x)     +f_{22}(x)\lambda x^{m-\deg g_{12}(x)}g^*_{12}(x)\equiv 0\pmod{x^m-\lambda^{-1}}\\
       \\
     \iff & x^{\deg g_{12}(x)}f_{21}(x) g^*_{11}(x)     +x^{\deg g_{11}(x)}f_{22}(x) g^*_{12}(x)\equiv 0\pmod{x^m-\lambda^{-1}}. 
    \end{split}
\end{equation*}
This completes the proof.
\end{proof}

In the following example, we illustrate the construction of quantum codes from quasi-twisted codes. All computations are performed by MAGMA\cite{magma}, and obtained codes are compared with Grassl's code table \cite{codetable}. 
\begin{example}
  Let $\mathbb{F}_q=\mathbb{F}_4=\mathbb{F}_2[w]$ with $w^2+w+1=0,\ \lambda=w,\  m=5$.
Let $C_1\subseteq R^2$ be a $w$-quasi-twisted code of length $2m=10$ and index $2$, generated by elements $g_1=(g_{11}(x),g_{12}(x))$ and $g_2=(0,g_{22}(x))$, where 
\begin{align*}
    g_{11}(x)=(x^2+wx+w), g_{12}(x)=w^2(x^2+wx+w), \text{ and } g_{22}(x)=(x^2+x+w)(x^2+wx+w).
\end{align*}
 Let $C_2\subseteq R^*{^2}$ is a $w^2$-quasi-twisted code of length $2m=10$ and index $2$, generated by elements $f_1=(f_{11}(x),0)$ and $f_2=(f_{21}(x),f_{22}(x))$, where
 \begin{align*}
  f_{11}(x)=x^5-w^2,\  f_{21}(x)=(x+w)(wx^3+wx^2+w^2x+w^2), \text{ and } f_{22}(x)=w(x+w).   
 \end{align*}
It is easy to see that
\begin{align*}
    g^*_{11}(x)f_{11}(x)\equiv 0\pmod{x^m-\lambda^{-1}},\\
 g_{22}^*(x)f_{22}(x)\equiv 0\pmod{x^m-\lambda^{-1}},\\
 x^{\deg g_{12}(x)}f_{21}(x) g^*_{11}(x)     +x^{\deg g_{11}(x)}f_{22}(x) g^*_{12}(x)\equiv 0\pmod{x^m-\lambda^{-1}}.
\end{align*}
Thus, by Theorem \ref{c1containc2e}, $C_1\subseteq C_2^{\perp_e}$. Using CSS construction, we get $[[10,2,4]]_4$, the same parameters as the best-known quantum error-correcting code \cite{codetable}.
\end{example}

\begin{remark}
Let $f(x) \in \mathbb{F}_q[x]$ be a polynomial of degree less than $m$. Using direct computation, evaluating the reciprocal polynomial $f^*(x)$ and its conjugate $\overline{f^*}(x)$ requires $O(m)$ operations. Furthermore, classical algorithms for the multiplication and division of two such polynomials have a time complexity of $O(m^2)$. Since constructing the Euclidean, Hermitian, and symplectic duals necessitates polynomial multiplication and division involving these (conjugate) reciprocals, the overall computational complexity using classical arithmetic is bounded by $O(m^2)$. However, asymptotically faster algorithms for polynomial arithmetic can optimize this process. For instance, employing the Karatsuba algorithm reduces the complexity of multiplication to approximately $O(m^{1.585})$. Consequently, utilizing such sub-quadratic algorithms for the underlying polynomial multiplication and division will yield a corresponding reduction in the overall complexity of computing these duals.
\end{remark}

In Table \ref{tab1} and Table \ref{tab2}, we present several quasi-twisted codes over ternary and quaternary fields with parameters same as the optimal or BKLC in \cite{codetable}. 
\begin{table}[h!]
    \centering
 \begin{tabular}{|c|c|c|c|c|c|}
    \hline
    $m$ & $g_{11}(x)$ &$g_{12}(x)$&$g_{22}(x)$& \text{Parameters} &\text{Remark}\\
         \hline
        $11$ & $x+1$ & $ x^2+x$ & \makecell{$x^5 + 2x^3 + 2x^2$\\ $+ 2x + 1$} & $[22,16,4]$& Optimal \\
         \hline
       $13$  & $x+1$ & $x^4 + x^3 + x^2 + 2$ &\makecell{$(x^3 + 2x + 1)$\\$(x^3 + x^2 + 2x + 1)$\\$(x^3 + 2x^2 + 1)$}  & $[26,16,6]$& BKLC\\
         \hline
        $13$ & $1$& \makecell{$x^5 + x^4 + x^3 $\\$+ 2x^2 + x + 1$}  &  \makecell{$(x^3 + 2x + 1)$\\$(x^3 + x^2 + 2x + 1)$\\$(x^3 + 2x^2 + 1)$}&  $[26,17,6]$& Optimal\\
         \hline
        $20$ & \makecell{$x^4 + 2x^3$\\$ + x^2 + 1$} & \makecell{$2x^5 + x^4 + x^3$\\$ + 2x^2 + 1$} & \makecell{$(x^2 + 2x + 2)$\\$(x^4 + x^2 + x + 1)$\\$(x^4 + x^2 + 2x + 1)$\\$(x^4 + x^3 + x^2 + 1)$} & $[40,22,9]$ & BKLC \\
         \hline
       $22$  &$1$ &$x^4 + 2x^3 + x^2 + x + 1$  & \makecell{$(x^{10} + 2x^6 + 2x^4 $\\$+ 2x^2 + 1)$\\$(x^2+1)$} & $[44,32,6]$ &BKLC\\
       \hline
    \end{tabular}
    \caption{Some examples of $\lambda=2$-quasi-twisted ternary code with optimal/BKLC parameters }
    \label{tab1}
\end{table}

\begin{table}[h!]

    \begin{tabular}{|c|c|c|c|c|c|}
    \hline
    $m$ & $g_{11}(x)$ &$g_{12}(x)$&$g_{22}(x)$& \text{Parameters} &\text{Remark}\\
              \hline
      $13$  & $x+w$  &$w^2x^5 + wx^3 + x + w$    &\makecell{$(x^6 + w^2x^5 + w^2x^3 $\\$+ x + 1)$\\$(x+w)$}   & $[26,18,6]$   &Optimal\\
        \hline
      $17$  & $1$  & $x+1$    & \makecell{$x^4 + x^3 + wx^2 $\\$+ wx + w^2)$}  & $[34,30,3]$    & Optimal\\
        \hline
      $19$  & $x+w$  &$wx^5 + wx^4 + wx^2 + 1$    &  \makecell{$x^9 + x^8 + w^2x^6 $\\$+ x^5 + x^4 + wx^3 $\\$+ x + 1$} &  $[38,28,6]$   &BKLC\\
        \hline
     $21$   & $1$  & $ wx^2 + x + 1$   & \makecell{$(x^3 + x^2 + x + w)$\\$(x^3+w)$}  &  $[42,36,4]$  & Optimal\\
        \hline
      $23$  & $1$ & $x^4 + wx^3 + x^2 + x + 1$   & \makecell{$x^11 + wx^9 + w^2x^7 $\\$+ wx^6 + x^5 + w^2x + w$}  & $[46,35,6]$   &BKLC\\
        \hline
      $23$  & $x+w^2$  & \makecell{$w^2x^5 + wx^3 + wx^2 $\\$+ x + w^2$}    & \makecell{$x^{11} + wx^9 + w^2x^7 $\\$+ wx^6 + x^5 + w^2x + w$}  &  $[36,34,7]$  &BKLC\\
        \hline
    \end{tabular}
    \caption{Some examples of $w$-quasi-twisted code over $\mathbb{F}_4=\mathbb{F}_2[w]$ with optimal/BKLC parameter,  where $ w^2+w+1=0$. }
    \label{tab2}
\end{table}

\section{ Quasi-twisted codes over $\mathbb{F}_q$ of index $2$ and additive constacyclic codes over $\mathbb{F}_{q^2}$} \label{qtl2ca}

In this section, we establish a correspondence between quasi-twisted codes of length $2m$ and index $2$ over $\mathbb{F}_q$, and additive constacyclic codes of length $m$ over $\mathbb{F}_{q^2}$. Additionally, we demonstrate that determining the trace Euclidean and trace Hermitian dual of additive constacyclic codes is the same as determining the Euclidean and symplectic dual of the corresponding quasi-twisted code. Throughout this section, we denote 
  $R=\mathbb{F}_q[x]/\langle x^m-\lambda\rangle$, $R^*=\mathbb{F}_q[x]/\langle x^m-\lambda^{-1}\rangle$, $R_A=\mathbb{F}_{q^2}[x]/\langle x^m-\lambda \rangle$, $R_A^*=\mathbb{F}_{q^2}[x]/\langle x^m-\lambda^{-1} \rangle$ and $C_A$ is an additive $\lambda$-constacyclic code over $\mathbb{F}_{q^2}$. Let $\{w_1,w_2\}$ be a basis of $\mathbb{F}_{q^2}$ over $\mathbb{F}_q$. We write an element $f(x)$ in $\mathbb{F}_{q^2}[x]$ (resp. in $R_A$) as $f(x)=w_1f_1(x)+w_2f_2(x)$, where $f_1(x),f_2(x)$ in $\mathbb{F}_q[x]$ (resp. in $R$).
\begin{definition}
    An additive $\lambda$-constacyclic code  over $\mathbb{F}_{q^2}$ with length $m$ is an $R$-submodule of $R_A$.
\end{definition}
Let $\{w_1,w_2\}$ be a basis of $\mathbb{F}_{q^2}$ over $\mathbb{F}_q$. Define a map $\phi: R^2\to R_A$  as follows:
$$\phi(f(x),g(x))=w_1f(x)+w_2g(x).$$

\begin{lemma}\label{qttoac}
    $R^2=(\mathbb{F}_q[x]/\langle x^m-\lambda\rangle )^2$ is isomorphic to $R_A=\mathbb{F}_{q^2}[x]/\langle x^m-\lambda \rangle$ as an $R$-module.
\end{lemma}
By the above lemma,  it is straightforward to conclude that $\phi(C)$ is an additive $\lambda$-constacyclic code of length $m$ over $\mathbb{F}_{q^2}$ if and only if $C$ is a $\lambda$-quasi-twisted code of length $2m$ and index $2$ over $\mathbb{F}_q$.
\begin{remark}
  If $C$ has parameters $[2m,k,d_s=d]_q$ then $\phi(C)$ will have parameters $(m,q^k,d_H=d)_{q^2}$, where $d_s$ represents the symplectic distance of $C$, defined by
  $$d_s(C)=\min \{wt_s((c(x),d(x))\mid (0,0)\neq (c(x),d(x))\in C\} $$ and
  $$wt_s((c(x),d(x)))=\mid \{0\leq i\leq m-1\mid (c_i,d_{i})\neq (0,0)\}\mid.$$
\end{remark}

\begin{theorem}\label{stca}
If $C_A$ is an additive $\lambda$-constacyclic code of length $m$ over $\mathbb{F}_{q^2}$ then there exist polynomials $g_{11}(x),g_{12}(x),g_{22}(x)\in R$ such that  
$g_{11}(x),g_{22}(x)|(x^m-\lambda)$, $\deg g_{12}(x)<\deg g_{22}(x)$ and $C_A=\langle w_1g_{11}(x)+w_2g_{12}(x),w_2g_{22}(x)\rangle$ with $\dim_{\mathbb{F}_q}(C_A)=2m-\deg g_{11}(x)-\deg g_{22}(x)$.     
\end{theorem}
\begin{proof}
 Let $C$ be a $\lambda$-quasi-twisted code of length $2m$ and index $2$ over $\mathbb{F}_q$ such that $\phi(C)=C_A$. The rest of the proof follows from  Theorem \ref{qtcode} and Lemma \ref{qttoac}. 
\end{proof}

\begin{remark}
The structure of additive cyclic codes over $\mathbb{F}_4$, $\mathbb{F}_{p^2}$ and $\mathbb{F}_{q^2}$ have been derived in \cite[Theorem 3.2]{shi2022}, \cite[Theorem 3.6 (ii)]{reza2025}, \cite[Theorem 3.2]{Verma2024}. These characterizations can be easily derived from Theorem \ref{stca} by choosing a suitable basis and $\lambda=1$. For instance, let $q=2$ and  $\gcd(m,q)=1$, then by Remark \ref{remgcd1} $g(x)=\gcd(g_{11}(x),g_{22}(x))$ divide $g_{12}(x)$. Let $g_{11}(x)=g(x)g_{11}'(x)$, $g_{22}(x)=g(x)g_{22}'(x)$ and $g_{12}(x)=g(x)g_{12}'(x)$. Take $w_1=w$ and $w_2=1$,  then additive cyclic codes over $\mathbb{F}_4=\mathbb{F}_2[w]$ is given by
$$C_A=\langle wg(x)g_{11}'(x)+g(x)g_{12}'(x),g(x)g_{22}'(x)\rangle.$$
\end{remark}

We denote $\text{Tr}$, the trace map on $\mathbb{F}_{q^2}$ over $\mathbb{F}_q$, defined as $\text{Tr}(\alpha)=\alpha+\alpha^q$ for all $\alpha\in \mathbb{F}_{q^2}$. The conjugate of an element $\alpha\in \mathbb{F}_{q^2}$ is defined as $\overline{\alpha}=\alpha^q$. We define the conjugate of a polynomial by taking the conjugate of coefficients. Let $\{w_1,w_2\}$ be a basis of $\mathbb{F}_{q^2}$ over $\mathbb{F}_q$ and  $g(x)=w_1g_1(x)+w_2g_2(x)$, $f(x)=w_1f_1(x)+w_2f_2(x)$ in $\mathbb{F}_{q^2}[x]$, where $f_i(x),g_i(x)\in \mathbb{F}_q[x]$ with degree at most $m-1$ for $i=1,2$. Then one can define the following bilinear forms over $\mathbb{F}_q$:
 \begin{enumerate}
     \item trace Euclidean: $\langle g(x),f(x)\rangle_{te}= \text{Tr}(\langle g(x),f(x)\rangle_e)$;
     \item trace Hermitian: $\langle g(x),f(x)\rangle_{th}= \text{Tr}(\langle g(x),f(x)\rangle_h)$;
     \item  symplectic: $\langle g(x),f(x)\rangle_{s}= \langle g_1(x),f_2(x)\rangle_e-\langle g_2(x),f_1(x)\rangle_e$,
 \end{enumerate}
where $\langle\ ,\  \rangle_e$ is the Euclidean inner product defined in Section \ref{dualcodes} and $\langle\ ,\  \rangle_h$ is the Hermitian inner product defined in Section \ref{dualcodesh}. 
If $g(x)$ has degree greater than or equal to $m$, then we reduce it modulo $x^m -\lambda$.
If $f(x)$ has degree greater than or equal to $m$, then we reduce it modulo $x^m -\lambda^{-1}$.
Hence, the left component is considered as an element of $R_A$, and the right component
is considered as an element of $R_A^*$. Let $C_A\subseteq R_A$ be an additive $\lambda$-constacyclic code of length $m$ over $\mathbb{F}_{q^2}$.  Then the dual codes are defined as

 \begin{enumerate}
	\item  trace Euclidean dual: 
	$C_A^{\perp_{te}}=\{f(x)\in R_A^*| \ \ \langle c(x),f(x)\rangle_{te}=0, \forall c(x)\in C_A\}.$
	\item trace Hermitian dual: 
$C_A^{\perp_{th}}=\{f(x)\in R_A^*| \ \ \langle c(x),f(x)\rangle_{th}=0, \forall c(x)\in C_A\}.$
	\item  symplectic dual:
	$C_A^{\perp_{s}}=\{f(x)\in R_A^*| \ \ \langle c(x),f(x)\rangle_{s}=0, \forall c(x)\in C_A\}.$
\end{enumerate} 
The dual code $C^{\perp_\#}_A\subseteq R_A^*$ is an additive $\lambda^{-1}$-constacyclic code of length $m$ over $\mathbb{F}_{q^2}$, where $\#$ represents the trace Euclidean, trace Hermitian, and symplectic dual.

\begin{remark}\label{eqsym}
    Note that the symplectic inner product depends on the choice of the basis. For a fixed basis $\{w_1,w_2\}$ of $\mathbb{F}_{q^2}$ over $\mathbb{F}_q$, by definition of symplectic inner product, we have $$C_A^{\perp_s}=\{w_1f_1(x)+w_2f_2(x)|\ (f_1(x),f_2(x))\in C^{\perp_s}\},$$
where $C$ is a $\lambda$-quasi-twisted code of length $2m$ and index $2$ over $\mathbb{F}_q$ and $C_A$ is an additive $\lambda$-constacyclic code of length $m$ over $\mathbb{F}_{q^2}$ such that  $\phi(C)=C_A$.   
\end{remark}

From the above remark, we have the following.

\begin{corollary}\label{a-symplectic-dual}
    Let $C_A=\langle w_1g_{11}(x)+w_2g_{12}(x),w_2g_{22}(x)\rangle$ be an additive $\lambda$-constacyclic code of length $m$ over $\mathbb{F}_{q^2}$,  satisfying Eq. \ref{eqqt}. Then
    $$C_A^{\perp_s}=\left \langle w_1\frac{x^m-\lambda^{-1}}{g_{22}^*(x)}+w_2 \lambda x^{m-\deg g_{12}+\deg g_{11}}\frac{(x^m-\lambda^{-1})g_{12}^*(x)}{g_{11}^*(x)g_{22}^*(x)}, w_2\frac{x^m-\lambda^{-1}}{g_{11}^*(x)}\right \rangle.$$
\end{corollary}
\begin{proof}
    Let $C$ be a $\lambda$-quasi-twisted code of length $2m$ and index $2$ over $\mathbb{F}_q$ such that $\phi(C)=C_A$. Then $C$ is generated by $(g_{11}(x),g_{12}(x))$ and $(0,g_{22}(x))$. The rest of the proof follows from Theorem \ref{thmsdual} and Remark \ref{eqsym}.
\end{proof}
The following corollary shows that Theorem \ref{selfsym} generalizes the result for symplectic self-orthogonality of additive cyclic codes over $\mathbb{F}_{p^2}$ obtained in \cite[Theorem 4.2]{reza2025}.  

\begin{corollary}\cite[Theorem 4.2]{reza2025}
Let $\{w_1=w,w_2=1\}$ be a basis for $\mathbb{F}_{p^2}=\mathbb{F}_p[w]$ over $\mathbb{F}_p$. Let $C_A=\langle wg_2(x)+g_1(x),h(x)\rangle$ be an additive cyclic code ($\lambda=1$) of length $m$ over $\mathbb{F}_{p^2}$ satisfying Eq. \ref{eqqt}. Then $C_A$ is self-orthogonal with respect to the symplectic inner product if and only if  \\
1) $g_{2}(x)h(x^{-1})\equiv 0\pmod{x^m-1}$,\\
   2)  $ g_{1}(x)g_{2}(x^{-1}) \equiv g_{2}(x)g_{1}(x^{-1})  \pmod{x^m-1}$.
\end{corollary}
\begin{proof}
  Let $C$ be a quasi-cyclic code (i.e. $\lambda=1$) of length $2m$ and index $2$ over $\mathbb{F}_p$, generated by $(g_2(x),g_1(x))$ and $(0,h(x)$. Then, by Lemma \ref{qttoac}, $\phi(C)=C_A$. Thus, by Remark \ref{eqsym}, $C_A$ is self-orthogonal if and only if $C$ is self-orthogonal. 
    By Theorem \ref{selfsym}, $C$ is self-orthogonal if and only if \\
    1) $g_{2}(x)h^*(x)\equiv 0\pmod{x^m-1}$,\\
   2)  $x^{\deg g_2(x)} g_2(x)g_{1}^*(x)-x^{\deg g_{1}(x)}g_{1}(x)g_{2}^*(x)\equiv 0\pmod{x^m-1}$.\\
  Equivalently,\\
   1) $g_{2}(x)h(x^{-1})\equiv 0\pmod{x^m-1}$,\\
   2)  $x^{\deg g_2(x)+\deg g_1(x)} g_2(x)g_{1}(x^{-1})\equiv x^{\deg g_{1}(x)+\deg g_2(x)}g_{1}(x)g_{2}(x^{-1})\pmod{x^m-1}$.\\
   This completes the proof. 
\end{proof}

Let $W=\{w_i\}_{i=1}^l$ be a basis for $\mathbb{F}_{q^l}$ over $\mathbb{F}_q$. Then the basis $W$ is called a trace orthogonal basis if $\text{Tr}(w_iw_j)=0$ for $i\neq j$. The basis $W$  is said to be a self-dual basis (or a trace orthonormal basis) if $W$ is a trace orthogonal basis and $\text{Tr}(w_i^2)=1$ for $1\leq i\leq l$. Also, $W$  is said to be an almost self-dual if $W$ is trace-orthogonal and $\text{Tr}(w_i^2)=1$ with possibly one exception (for details, see \cite{seroussi1980,jungickle1990}).  

\begin{lemma}\cite{jungickle1990,seroussi1980} \label{tebasis}   For every prime power $q$, there exists an almost self-dual basis of $\mathbb{F}_{q^l}$ over $\mathbb{F}_q$. Furthermore, $\mathbb{F}_{q^l}$ has a self-dual basis over $\mathbb{F}_q$ if and only if $q$ is even or  both $q$ and $l$  are odd.  
\end{lemma}

\begin{theorem} \label{oddte}
 Let $q$ be odd. Let $\{w_1,w_2\}$ be an almost self-dual basis for $\mathbb{F}_{q^2}$ over $\mathbb{F}_q$ 
 (that is,  $\text{Tr}(w_1w_2)=0$, $\text{Tr}(w_1^2)=1$ and  $\text{Tr}(w_2^2)=\epsilon$ is a non-square element in $\mathbb{F}_q$). 
 Then
 $$\langle g(x), f(x)\rangle_{te}=\langle g_1(x),f_1(x)\rangle_e+\epsilon \langle g_2(x),f_2(x)\rangle_e$$
 for all $g(x)=w_1g_1(x)+w_2g_2(x)$ and $f(x)=w_1f_1(x)+w_2f_2(x)$. 
 Moreover, in this case, we have  $$C_A^{\perp_{te}}=\{w_1f_1(x)+\epsilon^{-1}w_2f_2(x)\ | (f_1(x),f_2(x))\in C^{\perp_e}\},$$
where $C$ is a $\lambda$-quasi-twisted code of length $2m$ and index $2$ over $\mathbb{F}_q$ and $C_A$ is an additive $\lambda$-constacyclic code of length $m$ over $\mathbb{F}_{q^2}$ such that  $\phi(C)=C_A$.
\end{theorem}

\begin{proof}
    By definition of trace Euclidean inner product, we have 
    \begin{equation*}
        \begin{split}
         \langle g(x),f(x)\rangle_{te}&=\text{Tr}(\langle w_1g_1(x)+w_2g_2(x),w_1f_1(x)+w_2f_2(x)\rangle_e)\\
         &=\text{Tr}(\langle w_1g_1(x),w_1f_1(x)\rangle_e +\langle w_1g_1(x),w_2f_2(x)\rangle_e\\
    &+\langle w_2g_2(x),w_1f_1(x)\rangle_e+\langle w_2g_2(x),w_2f_2(x)\rangle_e)\\
    &=\text{Tr}(w_1^2)\langle g_1(x),f_1(x)\rangle_e +\text{Tr}(w_1w_2)\langle g_1(x),f_2(x)\rangle_e\\
    &+\text{Tr}(w_1w_2)\langle g_2(x),f_1(x)\rangle_e+\text{Tr}(w_2^2)\langle g_2(x),f_2(x)\rangle_e\\
    &=\langle g_1(x),f_1(x)\rangle_e+\epsilon\langle g_2(x),f_2(x)\rangle_e.
        \end{split}
    \end{equation*}

   Let $C$ be a $\lambda$-quasi-twisted code of length $2m$ and index $2$ and  $C_A$ be an additive $\lambda$-constacyclic code of length $m$ over $\mathbb{F}_{q^2}$ such that $\phi(C)=C_A$. Suppose 
   $$S=\{w_1f_1(x)+\epsilon^{-1}w_2f_2(x)\ | (f_1(x),f_2(x))\in C^{\perp_e}\}.$$   
It is easy to see that $|S|=|C^{\perp_e}|=|C_A^{\perp_{te}}|$. We show that $S=C_A^{\perp_{te}}$.  Let $f=(f_1(x),f_2(x))\in C^{\perp_e}$. Then $\langle c,f\rangle_e=0$ for all $c=(c_1(x),c_2(x))\in C$. Consequently,  $\phi(c)=w_1c_1(x)+w_2c_2(x)\in C_A$ and 
   \begin{equation*}
       \begin{split}
        \langle  w_1c_1(x)+w_2c_2(x), w_1f_1(x)+\epsilon^{-1}w_2f_2(x)\rangle_{te}&=\langle c_1(x),f_1(x)\rangle_e+\epsilon \langle c_2(x),\epsilon^{-1}f_2(x)\rangle_e\\
        &= \langle c_1(x),f_1(x)\rangle_e+\langle c_2(x),f_2(x)\rangle_e\\
        &=\langle c,f\rangle_e=0.
       \end{split}
   \end{equation*}
   Thus, $w_1f_1(x)+\epsilon^{-1}w_2f_2(x)\in C_A^{\perp_{te}}$, that is, $S\subseteq C_A^{\perp_{te}}$. Hence $S=C_A^{\perp_{te}}$ (since $|S|=|C_A^{\perp_{te}}|$).
\end{proof}

\begin{corollary}\label{teodddual}
    Let  $q$ be odd. Let $\{w_1,w_2\}$ be an almost self-dual basis for $\mathbb{F}_{q^2}$ over $\mathbb{F}_q$ 
 (that is,  $\text{Tr}(w_1w_2)=0$, $\text{Tr}(w_1^2)=1$ and  $\text{Tr}(w_2^2)=\epsilon$ is a non-square element in $\mathbb{F}_q$). 
 Let $C_A=\langle w_1g_{11}(x)+w_2g_{12}(x),w_2g_{22}(x)\rangle$ be an additive $\lambda$-constacyclic code of length $m$ over $\mathbb{F}_{q^2}$  satisfying Eq. \ref{eqqt}. Then
    $$C_A^{\perp_{te}}=\left \langle w_1\frac{x^m-\lambda^{-1}}{g^*_{11}(x)}, w_1\frac{-\lambda x^{m+\deg g_{11}(x)-\deg g_{12}(x)}(x^m-\lambda^{-1})g^*_{12}(x)}{g^*_{11}(x)g^*_{22}(x)}+\epsilon^{-1}w_2\frac{x^m-\lambda^{-1}}{g^*_{22}(x)}\right \rangle.$$ 
\end{corollary}
\begin{proof}
    Let $C$ be a $\lambda$-quasi-twisted code of length $2m$ and index $2$ over $\mathbb{F}_q$ such that $\phi(C)=C_A$. Then $C$ is generated by $(g_{11}(x),g_{12}(x))$ and $(0,g_{22}(x))$. The rest of the proof follows from Theorem \ref{thmedual} and Theorem \ref{oddte}.
\end{proof}

\begin{theorem}\label{qevente}
Let $q$ be even. Let $\{w_1,w_2\}$ be a self-dual basis for $\mathbb{F}_{q^2}$ over $\mathbb{F}_q$. Then
 $$\langle g(x),f(x)\rangle_{te}=\langle g_1(x),f_1(x)\rangle_e+ \langle g_2(x),f_2(x)\rangle_e$$
  for all  $g(x)=w_1g_1(x)+w_2g_2(x)$ and $f(x)=w_1f_1(x)+w_2f_2(x)$. Furthermore, in this case, we have  $$C_A^{\perp_{te}}=\{w_1f_1(x)+w_2f_2(x)\ | (f_1(x),f_2(x))\in C^{\perp_e}\},$$
where $C$ is a $\lambda$-quasi-twisted code of length $2m$ and index $2$ over $\mathbb{F}_q$ and $C_A$ is an additive $\lambda$-constacyclic code of length $m$ over $\mathbb{F}_{q^2}$ such that  $\phi(C)=C_A$.
\end{theorem}
\begin{proof}
 By definition of trace Euclidean inner product, we have 
    \begin{equation*}
        \begin{split}
         \langle g(x),f(x)\rangle_{te}&=\text{Tr}(\langle w_1g_1(x)+w_2g_2(x),w_1f_1(x)+w_2f_2(x)\rangle_e)\\
         &=\text{Tr}(\langle w_1g_1(x),w_1f_1(x)\rangle_e +\langle w_1g_1(x),w_2f_2(x)\rangle_e\\
    &+\langle w_2g_2(x),w_1f_1(x)\rangle_e+\langle w_2g_2(x),w_2f_2(x)\rangle_e)\\
    &=\text{Tr}(w_1^2)\langle g_1(x),f_1(x)\rangle_e +\text{Tr}(w_1w_2)\langle g_1(x),f_2(x)\rangle_e\\
    &+\text{Tr}(w_1w_2)\langle g_2(x),f_1(x)\rangle_e+\text{Tr}(w_2^2)\langle g_2(x),f_2(x)\rangle_e\\
    &=\langle g_1(x),f_1(x)\rangle_e+\langle g_2(x),f_2(x)\rangle_e.
        \end{split}
    \end{equation*}
   Let $C$ be a $\lambda$-quasi-twisted code of length $2m$ and index $2$ and  $C_A$ be an additive $\lambda$-constacyclic code of length $m$ over $\mathbb{F}_{q^2}$ such that $\phi(C)=C_A$. Suppose 
   $$S=\{w_1f_1(x)+w_2f_2(x)\ | (f_1(x),f_2(x))\in C^{\perp_e}\}.$$   
It is easy to see that $|S|=|C^{\perp_e}|=|C_A^{\perp_{te}}|$. We show that $S=C_A^{\perp_{te}}$.  Let $f=(f_1(x),f_2(x))\in C^{\perp_e}$. Then $\langle c,f\rangle_e=0$ for all $c=(c_1(x),c_2(x))\in C$. Consequently,  $\phi(c)=w_1c_1(x)+w_2c_2(x)\in C_A$ and 
   \begin{equation*}
       \begin{split}
        \langle  w_1c_1(x)+w_2c_2(x), w_1f_1(x)+w_2f_2(x)\rangle_{te}&=\langle c_1(x),f_1(x)\rangle_e+ \langle c_2(x),f_2(x)\rangle_e\\
        &= \langle c_1(x),f_1(x)\rangle_e+\langle c_2(x),f_2(x)\rangle_e\\
        &=\langle c,f\rangle_e=0.
       \end{split}
   \end{equation*}
   Thus, $w_1f_1(x)+w_2f_2(x)\in C_A^{\perp_{te}}$, that is, $S\subseteq C_A^{\perp_{te}}$. This completes the proof. 
\end{proof}

\begin{corollary}\label{tedual}
    Let $q$ be even and let $\{w_1,w_2\} $ be a self-dual basis. Let $C_A=\langle w_1g_{11}(x)+w_2g_{12}(x),w_2g_{22}(x)\rangle$ be an additive $\lambda$-constacyclic code of length $m$ over $\mathbb{F}_{q^2}$ satisfying Eq. \ref{eqqt}. Then
    $$C_A^{\perp_{te}}=\left \langle w_1\frac{x^m-\lambda^{-1}}{g^*_{11}(x)},w_1\frac{-\lambda x^{m+\deg g_{11}(x)-\deg g_{12}(x)}(x^m-\lambda^{-1})g^*_{12}(x)}{g^*_{11}(x)g^*_{22}(x)}+w_2\frac{x^m-\lambda^{-1}}{g^*_{22}(x)}\right \rangle.$$
\end{corollary}
\begin{proof}
    Let $C$ be a $\lambda$-quasi-twisted code of length $2m$ and index $2$ over $\mathbb{F}_q$ such that $\phi(C)=C_A$. Then $C$ is generated by $(g_{11}(x),g_{12}(x))$ and $(0,g_{22}(x))$. The rest of the proof follows from Theorem \ref{thmedual} and Theorem \ref{qevente}.
\end{proof}

Next, we recall some results associated with the forms on $\mathbb{F}_{q^2}$ over $\mathbb{F}_q$. We utilize these results to describe the trace Hermitian inner product. Define 
$\mathcal{B}_2:\mathbb{F}_{q^2}^2\to \mathbb{F}_q$ as $$\mathcal{B}_2(\alpha,\beta)=\text{Tr}(\alpha\beta^q)=\alpha\beta^q+\alpha^q\beta.$$
Observe that if $q$ is odd, $\mathcal{B}_2$ is a non-degenerate symmetric bilinear form over $\mathbb{F}_q$, and if $q$ is even, then $\mathcal{B}_2$ is a symplectic form (that is, non-degenerate, skew-symmetric and $\mathcal{B}_2(\alpha,\alpha)=0$ for all $\alpha\in \mathbb{F}_{q^2}$) over $\mathbb{F}_q$.

\begin{lemma}\cite{Artin1957}
\label{artin}
    If $q$ is odd, then there are exactly two equivalence classes of non-degenerate symmetric bilinear forms on $\mathbb{F}_{q^2}$ over $\mathbb{F}_q$, represented by the identity matrix $I$ and the matrix $M={\rm diag}(1,\eta)$, where $\eta$ is an arbitrary non-square element in $\mathbb{F}_q$.
\end{lemma}
We define 
 \begin{equation}
 \label{epsilon}
\epsilon = 
\left\{\begin{array}{rl}
1~, & \text{if}~ q \equiv 3 \pmod{4},\\
\eta~,& \text{if}~ q \equiv 1 \pmod{4}. 
\end{array}\right.
\end{equation}

 \begin{proposition}
    If $q$ is odd, then there exists a basis $\{w_1,w_2\}$ of $\mathbb{F}_{q^2}$ over $\mathbb{F}_q$ such that $\mathcal{B}_2(w_1,w_2)=\mathcal{B}_2(w_2,w_1)=0$, $\mathcal{B}_2(w_1,w_1)=1$ and $\mathcal{B}_2(w_2,w_2)=\epsilon$, where $\epsilon$ is defined in Eq. (\ref{epsilon}).  
 \end{proposition}
\begin{proof}
Let $\eta \in \mathbb{F}_q^*$ be a nonsquare. We construct a basis $\{w_1,w_2\}$ with Gram matrix $\text{diag}(1,-\eta)$. Since $X^2-\eta$ is irreducible, it has two roots $\pm\theta$. The Frobenius automorphism interchanges them, hence $\theta^q=-\theta$. 
Choose $w_1$ such that $N_{\mathbb{F}_{q^2}/\mathbb{F}_q}(w_1)=\frac{1}{2}$ and set $w_2=\theta w_1$. Then 
\begin{eqnarray*}
\mathcal{B}_2(w_1,w_1) &=& w_1w_1^q+w_1^qw_1 = 2 N_{\mathbb{F}_{q^2}/\mathbb{F}_q}(w_1)=1, \\
\mathcal{B}_2(w_1,w_2) &=& w_1(\theta w_1)^q+w_1^q(\theta w_1) = N_{\mathbb{F}_{q^2}/\mathbb{F}_q}(w_1)(\theta^q+\theta)=0, \\
\mathcal{B}_2(w_2,w_2) &=& \theta w_1(\theta w_1)^q+(\theta w_1)^q \theta w_1 = 2N_{\mathbb{F}_{q^2}/\mathbb{F}_q}(w_1)(-\theta^2)=-\eta. 
\end{eqnarray*}
If $q \equiv 1 \pmod{4}$, then $-1$ is a square in $\mathbb{F}_q$. Choose $c\in \mathbb{F}_q$ such that $c^2=-1$ and replace $w_2$ by $w'_2=cw_2$. 
If $q \equiv 3 \pmod{4}$, then $-\eta$ is a square in $\mathbb{F}_q$. Choose $c\in \mathbb{F}_q$ such that $c^2=-\eta$ and replace $w_2$ by $w'_2=c^{-1}w_2$. 
\end{proof}

 From the above proposition, when $q$ is odd, there exists a basis $\{w_1,w_2\}$ of $\mathbb{F}_{q^2}$ over $\mathbb{F}_q$ that satisfies the following
 \begin{equation}\label{basis2}
     \text{Tr}(w_1^{q}w_2)=0= \text{Tr}(w_1w_2^{q}), \text{Tr}(w_1^{q+1})=1,\text{ and Tr}(w_2^{q+1})=\epsilon,
 \end{equation}
where $\epsilon$ is defined in Eq. (\ref{epsilon}.

\begin{theorem} \label{oddth}
Let $q$ be odd and $\{w_1,w_2\}$ be a basis for $\mathbb{F}_{q^2}$ over $\mathbb{F}_q$ satisfying Eq. \ref{basis2}. Then
 $$\langle g(x),f(x)\rangle_{th}=\langle g_1(x),f_1(x)\rangle_e+\epsilon \langle g_2(x),f_2(x)\rangle_e$$
 for all $g(x)=w_1g_1(x)+w_2g_2(x)$ and $f(x)=w_1f_1(x)+w_2f_2(x)$. Moreover, in this case, we have  $$C_A^{\perp_{th}}=\{w_1f_1(x)+\epsilon^{-1}w_2f_2(x)\ | (f_1(x),f_2(x))\in C^{\perp_e}\},$$
where $C$ is a $\lambda$-quasi-twisted code of length $2m$ and index $2$ over $\mathbb{F}_q$ and $C_A$ is an additive $\lambda$-constacyclic code of length $m$ over $\mathbb{F}_{q^2}$ such that  $\phi(C)=C_A$.
\end{theorem}
\begin{proof}
     By definition of the trace Hermitian inner product, we have 
    \begin{equation*}
        \begin{split}
         \langle g(x),f(x)\rangle_{th}&=\text{Tr}(\langle g(x),f(x)\rangle_h)=\text{Tr}(\langle g(x),\overline{f}(x)\rangle_e)\\
         &=\text{Tr}(\langle w_1g_1(x)+w_2g_2(x),w_1^qf_1(x)+w_2^qf_2(x)\rangle_e)\\
         &=\text{Tr}(\langle w_1g_1(x),w_1^qf_1(x)\rangle_e +\langle w_1g_1(x),w_2^qf_2(x)\rangle_e\\
    &+\langle w_2g_2(x),w_1^qf_1(x)\rangle_e+\langle w_2g_2(x),w_2^qf_2(x)\rangle_e)\\
    &=\text{Tr}(w_1^{q+1})\langle g_1(x),f_1(x)\rangle_e +\text{Tr}(w_1w_2^q)\langle g_1(x),f_2(x)\rangle_e\\
    &+\text{Tr}(w_1^qw_2)\langle g_2(x),f_1(x)\rangle_e+\text{Tr}(w_2^{q+1})\langle g_2(x),f_2(x)\rangle_e\\
    &=\langle g_1(x),f_1(x)\rangle_e+\epsilon\langle g_2(x),f_2(x)\rangle_e.
        \end{split}
    \end{equation*}
   Let $C$ be a $\lambda$-quasi-twisted code of length $2m$ and index $2$ and  $C_A$ be an additive $\lambda$-constacyclic code of length $m$ over $\mathbb{F}_{q^2}$ such that $\phi(C)=C_A$. Let 
   $$S=\{w_1f_1(x)+\epsilon^{-1}w_2f_2(x)\ | (f_1(x),f_2(x))\in C^{\perp_e}\}.$$   
It is easy to see that $|S|=|C^{\perp_e}|=|C_A^{\perp_{th}}|$. We show that $S=C_A^{\perp_{th}}$.  Let $f=(f_1(x),f_2(x))\in C^{\perp_e}$. Then $\langle c,f\rangle_e=0$ for all $c=(c_1(x),c_2(x))\in C$. Consequently,  $w_1c_1(x)+w_2c_2(x)\in C_A$ and 
   \begin{equation*}
       \begin{split}
        \langle  w_1c_1(x)+w_2c_2(x),w_1f_1(x)+\epsilon^{-1}w_2f_2(x)\rangle_{th}&=\langle c_1(x),f_1(x)\rangle_e+\epsilon \langle c_2(x),\epsilon^{-1}f_2(x)\rangle_e\\
        &= \langle c_1(x),f_1(x)\rangle_e+\langle c_2(x),f_2(x)\rangle_e\\
        &=\langle c,f\rangle_e=0.
       \end{split}
   \end{equation*}
   Thus, $w_1f_1(x)+\epsilon^{-1}w_2f_2(x)\in C_A^{\perp_{th}}$, that is, $S\subseteq C_A^{\perp_{th}}$. Hence $S=C_A^{\perp_{th}}$ (since $|S|=|C_A^{\perp_{th}}|$).
\end{proof}

\begin{corollary}\label{thodddual}
    Let $q$ be odd. Let $\{w_1,w_2\}$ be a basis for $\mathbb{F}_{q^2}$ over $\mathbb{F}_q$ satisfying Eq. \ref{basis2}. Let $C_A=\langle w_1g_{11}(x)+w_2g_{12}(x),w_2g_{22}(x)\rangle$ be an additive $\lambda$-constacyclic code of length $m$ over $\mathbb{F}_{q^2}$ satisfying Eq. \ref{eqqt}. Then
    $$C_A^{\perp_{th}}=\left \langle w_1\frac{x^m-\lambda^{-1}}{g^*_{11}(x)}, w_1\frac{-\lambda x^{m+\deg g_{11}(x)-\deg g_{12}(x)}(x^m-\lambda^{-1})g^*_{12}(x)}{g^*_{11}(x)g^*_{22}(x)}+\epsilon^{-1}w_2\frac{x^m-\lambda^{-1}}{g^*_{22}(x)}\right \rangle.$$ 
\end{corollary}
\begin{proof}
    Let $C$ be a $\lambda$-quasi-twisted code of length $2m$ and index $2$ over $\mathbb{F}_q$ such that $\phi(C)=C_A$. Then $C$ is generated by $(g_{11}(x),g_{12}(x))$ and $(0,g_{22}(x))$. The rest of the proof follows from Theorem \ref{thmedual} and Theorem \ref{oddth}.
\end{proof}

Recall that, when $q$ is even, $\mathcal{B}_2$ is a symplectic form on $\mathbb{F}_{q^2}$ over $\mathbb{F}_q$. Therefore, we have the following.

\begin{theorem}\label{qeventh}
Let $q$ be even. Let $\{w_1,w_2\}$ be a basis of $\mathbb{F}_{q^2}$ over $\mathbb{F}_q$.  Then 
 $$\langle g(x),f(x)\rangle_{th} =\delta\langle g(x),f(x)\rangle_s,$$   for all $f(x)=w_1f_1(x)+w_2f_2(x)$ and $g(x)=w_1g_1(x)+w_2g_2(x)$, where $\delta=\text{Tr}(w_1w_2^q)\neq 0$.  Moreover, in this case, we have  $$C_A^{\perp_{th}}=\{w_1f_1(x)+w_2f_2(x)\ | (f_1(x),f_2(x))\in C^{\perp_s}\}$$
where $C$ is a $\lambda$-quasi-twisted code of length $2m$ and index $2$ over $\mathbb{F}_q$ and $C_A$ is an additive $\lambda$-constacyclic code of length $m$ over $\mathbb{F}_{q^2}$ such that  $\phi(C)=C_A$.
\end{theorem}
\begin{proof}
    By the definition of trace Hermitian inner product, we have
 \begin{equation*}
     \begin{split}
        \langle g(x),f(x)\rangle_{th} &=\text{Tr}(w_1^{q+1})\langle g_1(x),f_1(x)\rangle_e +\text{Tr}(w_1w_2^q)[\langle g_1(x),f_2(x)\rangle_e+
    \langle g_2(x),f_1(x)\rangle_e]\\
    &+\text{Tr}(w_2^{q+1})\langle g_2(x),f_2(x)\rangle_e. 
     \end{split}
 \end{equation*}
  Since $q$ is even, we have  $\text{Tr}(w_1^{q+1})=0$, $\text{Tr}(w_2^{q+1})=0$, and $\text{Tr}(w_1w_2^q)=\text{Tr}(w_1^qw_2)$. Consequently,  $$\langle g(x),f(x)\rangle_{th} =\text{Tr}(w_1w_2^q)[\langle g_1(x),f_2(x)\rangle_e+
    \langle g_2(x),f_1(x)\rangle_e]=\text{Tr}(w_1w_2^q)\langle g(x),f(x)\rangle_s.$$ 
  We claim that $\text{Tr}(w_1w_2^q)\neq0$. For this,  
  \begin{equation*}
      \begin{split}
          \text{Tr}(w_1w_2^q)=0 \iff& w_1w_2^q+w_1^qw_2=0 \iff w_1w_2(w_2^{q-1}+w_1^{q-1})=0\\
          \iff& w_2^{q-1}+w_1^{q-1}=0\iff w_2^{q-1}=w_1^{q-1}\\
         \iff& (w_2w_1^{-1})^{q-1}=1\iff w_2w_1^{-1}\in \mathbb{F}_q \ i.e. \ w_2=\alpha w_1
 \text{ for some } \alpha\in \mathbb{F}_q
 \end{split}
  \end{equation*}
  which is not true. Thus $\text{Tr}(w_1w_2^q)\neq0$.
  The rest of the proof follows similar to Theorem \ref{oddth}.
\end{proof}

\begin{corollary}\label{thdual}
    Let $q$ be even and let $\{w_1,w_2\} $ be a basis. Let $C_A=\langle w_1g_{11}(x)+w_2g_{12}(x),w_2g_{22}(x)\rangle$ be an additive $\lambda$-constacyclic code of length $m$ over $\mathbb{F}_{q^2}$ satisfying Eq. \ref{eqqt}. Then
    $$C_A^{\perp_{th}}=\left \langle  w_1\frac{x^m-\lambda^{-1}}{g_{22}^*(x)}+w_2\lambda x^{m-\deg g_{12}+\deg g_{11}}\frac{(x^m-\lambda^{-1})g_{12}^*(x)}{g_{11}^*(x)g_{22}^*(x)},  w_2\frac{x^m-\lambda^{-1}}{g_{11}^*(x)}\right \rangle.$$
\end{corollary}
\begin{proof}
    Let $C$ be a $\lambda$-quasi-twisted code of length $2m$ and index $2$ over $\mathbb{F}_q$ such that $\phi(C)=C_A$. Then $C$ is generated by $(g_{11}(x),g_{12}(x))$ and $(0,g_{22}(x))$. The rest of the proof follows from Theorem \ref{thmsdual} and Theorem \ref{qeventh}.
\end{proof}

\begin{remark}
Corollary \ref{tedual} and Corollary \ref{thdual} generalize the results obtained in \cite[Theorem 3.3 and Theorem 4.1]{shi2022} with an appropriate selection of the basis. Furthermore, in \cite{shi2022}, the authors characterized the duals of additive cyclic codes via the existence of certain polynomials (which are given implicitly), while in this work, we explicitly determine the polynomials that generate the duals.
\end{remark}

\begin{example}
Let $m=22$, $q=9$ and $\mathbb{F}_9=\mathbb{F}_3[w]$ with $w^2+2w+2=0$. Let $C$ be a $2$-quasi-twisted code of length $2m$ and index $2$ over $\mathbb{F}_3$ generated by $(g_{11}(x),g_{12}(x))$ and $(0,g_{22}(x))$, where
\begin{align*}
    g_{11}(x)=&x^2+1,\\
    g_{12}(x)=&x^4+x^3+x,\\
    g_{22}(x)=&x^{10}+2x^8+2x^6+2x^4+1.
\end{align*}
Let $B=\{2,w^2\}$ be basis of $\mathbb{F}_9$ over $\mathbb{F}_3$. Then by corresponding additive $2$-constacyclic code over $\mathbb{F}_9$ of length $m$ is given by 
\begin{align*}
   C_A&=\langle 2g_{11}(x)+w^2g_{12}(x),w^2g_{22}(x)\rangle\\
   &=\langle 2x^2+2+w^2x^4+w^2x^3+w^2x, w^2x^{10}+2w^2x^8+2w^2x^6+2w^2x^4+w^2\rangle.
\end{align*}
Note that 
\begin{align*}
     g_{11}^*(x)=&x^2+1,\\
    g_{12}^*(x)=&x^3+x+1,\\
    g_{22}^*(x)=&x^{10}+2x^6+2x^4+2x^2+1.
\end{align*}
By Corollary \ref{a-symplectic-dual}, the symplectic dual of $C_A$ is given by 
{\footnotesize \begin{align*}
    C_A^{\perp_s}=&\langle  2(x^{12}+x^8+x^6+2x^4+x^2+1)+ 2w^2x^{20} (x^{13}+x^{10}+x^9+2x^8+x^7+2x^6+2x^5+2x^4+x^3+x+1),\\ &w^2(x^{20}+2x^{18}+x^{16}+2x^{14}+x^{12}+2x^{10}+x^8+2x^6+x^4+2x^2+1)\rangle\\ & \\ 
    =& \langle 2(x^{12}+x^8+x^6+2x^4+x^2+1)+ w^6x^{19} + w^6x^{18} + w^2x^{17} + w^2x^{16} + w^6x^{15} + w^6x^{14} + w^2x^{13} +w^2x^{12} \\
    &+ w^6x^{10} + w^2x^9 + w^6x^8 + w^2x^6 + w^6x^5 + w^2x^3 + w^2x^2 + w^6x + w^2,\\
&w^2(x^{20}+2x^{18}+x^{16}+2x^{14}+x^{12}+2x^{10}+x^8+2x^6+x^4+2x^2+1)\rangle.
\end{align*}}
Let $\{ w^2,1\}$ be another basis of $\mathbb{F}_9$ over $\mathbb{F}_3$. Then $Tr(w^2\cdot w^2)=Tr(w^4)=1$, $Tr(w^2\cdot 1)=Tr(w^2)=0$, and $Tr(1)=2$ is non-square in $\mathbb{F}_3$.  By Corollary \ref{teodddual}, the trace Euclidean dual of $C_A$ is given by 
{\footnotesize \begin{align*}
 C_A^{\perp_{te}} =&\langle w^2(x^{20} + 2x^{18} + x^{16} + 2x^{14} + x^{12} + 2x^{10} + x^8 + 2x^6 + x^4 + 2x^2 + 1),\\
 &-2w^2x^{20} (x^{13}+x^{10}+x^9+2x^8+x^7+2x^6+2x^5+2x^4+x^3+x+1)+ 2(x^{12}+x^8+x^6+2x^4+x^2+1)\rangle\\
\end{align*}}
For trace Hermitian dual, let $\{w,w^3\}$ be a basis of $\mathbb{F}_9$ over $\mathbb{F}_3$. Observe that $Tr(w^3\cdot w)=Tr(w^4)=1$, $Tr(w^3\cdot w^3)=Tr(w\cdot w^9)=0$, and $Tr(w^3\cdot w^9)=1$. By Corollary \ref{thodddual}, the trace Hermitian dual is given by 
{\footnotesize \begin{align*}
    C_A^{\perp_{th}}=&\langle w(x^{20} + 2x^{18} + x^{16} + 2x^{14} + x^{12} + 2x^{10} + x^8 + 2x^6 + x^4 + 2x^2 + 1),\\
 &-2wx^{20} (x^{13}+x^{10}+x^9+2x^8+x^7+2x^6+2x^5+2x^4+x^3+x+1)+ w^3(x^{12}+x^8+x^6+2x^4+x^2+1)\rangle.\\
\end{align*}}
\end{example}

\section{Quasi-twisted codes over $\mathbb{F}_q$ of index $l$ and additive constacyclic codes over $\mathbb{F}_{q^l}$}\label{qttoconstal}
In this section, we extend the results established in Section \ref{qtl2ca} for arbitrary index $l$. In this section, we denote $R_A=\mathbb{F}_{q^l}[x]/\langle x^m-\lambda \rangle$, $R_A^*=\mathbb{F}_{q^l}[x]/\langle x^m-\lambda^{-1} \rangle$, $R=\mathbb{F}_q[x]/\langle x^m-\lambda \rangle$, $R^*=\mathbb{F}_q[x]/\langle x^m-\lambda^{-1} \rangle$ and $C_A$ is an additive $\lambda$-constacyclic code of length $m$ over $\mathbb{F}_{q^l}$. Let $\{w_i\}_{i=1}^l$ be a basis of $\mathbb{F}_{q^l}$ over $\mathbb{F}_q$, we write an element $f(x)$ in $\mathbb{F}_{q^l}[x]$  (resp. in $R_A$) as $f(x)=\sum_{i=1}^lw_if_i(x)$, where $f_i(x)$ in $\mathbb{F}_q[x]$ (resp. in $R$) for $1\leq i\leq l$. Throughout the section, we define trace map $\text{Tr}$ on $\mathbb{F}_{q^l}$ over $\mathbb{F}_q$ as $$\text{Tr}(\alpha)=\text{Tr}_{\mathbb{F}_{q^l}/\mathbb{F}_q}(\alpha)=\sum_{i=0}^{l-1} \alpha^{q^i}$$ for all $\alpha\in \mathbb{F}_{q^l}$. The definitions of the trace Euclidean and the trace Hermitian inner products are naturally extended on $R_A$. The conjugate of an element $\alpha\in \mathbb{F}_{q^l}$ is defined as $\overline{\alpha}=\alpha^{q^r}$, whenever $l=2r$ for some positive integer $r$.
\begin{definition}
    An additive $\lambda$-constacyclic code  over $\mathbb{F}_{q^l}$ of length $m$ is an $R$-submodule of $R_A$.
\end{definition}
Let $W=\{w_i\}_{i=1}^l$ be a basis for $\mathbb{F}_{q^l}$ over $\mathbb{F}_q$. Then similar to Lemma \ref{qttoac}, we conclude that  $\lambda$-quasi-twisted codes of length $lm$ with index $l$ over $\mathbb{F}_q$ are in one-to-one correspondence with additive $\lambda$-constacyclic codes of length $m$ over $\mathbb{F}_{q^l}$. Recall from Lemma \ref{tebasis} that there is an almost self-dual basis for $\mathbb{F}_{q^l}$ over $\mathbb{F}_q$ for every prime power $q$. That is, there exists a basis $\{w_i\}_{i=1}^l$ such that $\text{Tr}(w_i^2)=1$, for $1\leq i\leq l-1$, $\text{Tr}(w_iw_j)=0$ for all $i\neq j$ and $\text{Tr}(w_l^2)=\epsilon$. Furthermore, there is a self-dual basis $W=\{w_i\}_{i=1}^l$ for $\mathbb{F}_{q^l}$ over $\mathbb{F}_q$ if and only if either $q$ is even or $q$ and $l$ both are odd. Now, we describe the relationship between inner products and dual codes. The proofs will follow analogous to those presented in Section \ref{qtl2ca}.

\begin{theorem}
    Suppose $q$ is odd and $l$ is even. Let $\{w_i\}_{i=1}^l$ be an almost self-dual basis for $\mathbb{F}_{q^l}$ over $\mathbb{F}_q$ and $\epsilon=\text{Tr}(w_l^2)$ be a non-square element in $\mathbb{F}_q$. Then
 $$\langle g(x),f(x)\rangle_{te}=\langle g_1(x),f_1(x)\rangle_e+\dots+\langle g_{l-1}(x),f_{l-1}(x)\rangle_e+\epsilon \langle g_l(x),f_l(x)\rangle_e$$
 for all  $g(x)=\sum_{i=1}^lw_ig_i(x)$ and $f(x)=\sum_{i=1}^lw_if_i(x)$. Moreover, in this case, we have  $$C_A^{\perp_{te}}=\left \{\sum_{i=1}^{l-1}w_if_i(x)+\epsilon^{-1}w_lf_l(x)\ | (f_1(x),\dots,f_l(x)\in C^{\perp_e}\right \},$$
where $C$ is a $\lambda$-quasi-twisted code of length $lm$ and index $l$ over $\mathbb{F}_q$ and $C_A$ is an additive $\lambda$-constacyclic code of length $m$ over $\mathbb{F}_{q^l}$ such that  $\phi(C)=C_A$.
\end{theorem}

\begin{theorem}
     Suppose either $q$ is even or $q$ and  $l$ both are odd. Let $\{w_i\}_{i=1}^l$ be a self-dual basis for $\mathbb{F}_{q^l}$ over $\mathbb{F}_q$. Then
 $$\langle g(x),f(x)\rangle_{te}=\langle g_1(x),f_1(x)\rangle_e+\dots+\langle g_{l-1}(x),f_{l-1}(x)\rangle_e+ \langle g_l(x),f_l(x)\rangle_e$$
 for all $g(x)=\sum_{i=1}^lw_ig_i(x)$ and $f(x)=\sum_{i=1}^lw_if_i(x)$. Moreover, in this case, we have  $$C_A^{\perp_{te}}=\left \{\sum_{i=1}^{l}w_if_i(x)\ | (f_1(x),\dots,f_l(x)\in C^{\perp_e}\right \},$$
where $C$ is a $\lambda$-quasi-twisted code of length $lm$ and index $l$ over $\mathbb{F}_q$ and $C_A$ is an additive $\lambda$-constacyclic code of length $m$ over $\mathbb{F}_{q^l}$ such that  $\phi(C)=C_A$.
\end{theorem} 
Recall that the trace Hermitian inner product is defined over $\mathbb{F}_{q^l}$ only when $l$ is even. Thus, we consider two cases based on whether $q$ is odd or even with $l$ is even. Let $l=2r$, for some $r$. To describe the trace Hermitian case, we define a form $\mathcal{B}_l:\mathbb{F}_{q^l}^2\to \mathbb{F}_q$ as
$$\mathcal{B}_l(\alpha,\beta)=\text{Tr}(\alpha\beta^{q^r})=\text{Tr}_{\mathbb{F}_{q^l}/\mathbb{F}_q}(\alpha\beta^{q^r}).$$
Observe that when $q$ is odd, $\mathcal{B}_l$ is a non-degenerate symmetric bilinear form.  When $q$ is even, we have 
$$\mathcal{B}_l(\alpha,\alpha)=\text{Tr}_{\mathbb{F}_{q^l}/\mathbb{F}_q}(\alpha\alpha^{q^r})=0$$
for all $\alpha\in \mathbb{F}_{q^l}$. Thus  $\mathcal{B}_l$ is a symplectic form on $\mathbb{F}_{q^l}$ over $\mathbb{F}_q$ for even $q$. We have the following results analogous to Section \ref{qtl2ca}.

\begin{lemma}\cite{Artin1957}\label{h-l-add}
    If $q$ is odd, then there are exactly two equivalence classes of non-degenerate symmetric bilinear forms on $\mathbb{F}_{q^l}$ over $\mathbb{F}_q$, represented by the identity matrix $I$ and the matrix $M={\rm diag}(1,1,\dots,\eta)$, where $\eta$ is an arbitrary non-square element in $\mathbb{F}_q$.
\end{lemma}
We define 
 \begin{equation}
 \label{epsilon2}
\epsilon = 
\left\{\begin{array}{rl}
1~, & \text{if}~ q \equiv 3 \pmod{4}, r \ \text{odd}\\
\eta~,& \text{otherwise}. 
\end{array}\right.
\end{equation}

 \begin{proposition}
    If $q$ is odd and $l=2r$, then there exists a basis $\{w_i\}_{i=1}^l$ of $\mathbb{F}_{q^l}$ over $\mathbb{F}_q$ such that $\mathcal{B}_l(w_i,w_j)=\mathcal{B}_l(w_j,w_i)=0$ for $1\leq i<j\leq l$, $\mathcal{B}_l(w_i,w_i)=1$ for $1\leq i\leq l-1$ and $\mathcal{B}_l(w_l,w_l)=\epsilon$, where $\epsilon$ is defined in Eq. (\ref{epsilon2}).  
 \end{proposition}

\begin{proof}
Denote $F=\mathbb{F}_q$, $K=\mathbb{F}_{q^r}$, $L=\mathbb{F}_{q^{2r}}$. 
Let $\eta \in F^*$ be a nonsquare. We construct a basis $\{w_1,\dots,w_l\}$ with Gram matrix $\text{diag}(1,1,\dots,1,(-1)^r\eta)$.

Let $\rho  \in K$ be a nonsquare.
Since $X^2-\rho$ is irreducible, it has two roots $\pm\theta \in L$. The Frobenius automorphism $\sigma(\alpha)=\alpha^{q^r}$ interchanges them, hence $\sigma(\theta)=-\theta$. 
Then $L=K\oplus K\theta$. For $a,b,c,d\in K$, we have 
$$(a+b\theta)\sigma(c+d\theta)=(ac-\rho bd)+(bc-ad)\theta.$$
For $x,y\in K$, we have 
$$Tr_{L/K}(x+y\theta)=2x.$$
Therefore, 
$${\mathcal B}_l(a+b\theta,c+d\theta)= Tr_{K/F}(2ac-2\rho bd).$$
In particular, the $F$-subspaces $K$ and $K\theta$ are orthogonal. 
Choose an $F$-basis $e_1,\dots,e_r$ of $K$, and define a matrix 
$$
G=\left({Tr}_{K/F}(e_i e_j)\right)_{i,j=1}^r.
$$
It is clear that $G$ is invertible. 
Let $M_\rho$ be the matrix of multiplication by $\rho$ on $K$. Then 
$$\rho e_j = \sum_{k=1}^r(M_\rho)_{kj}e_k.$$
Therefore, 
$$\begin{aligned}
(GM_\rho)_{ij}
&=\sum_{k=1}^rG_{ik}(M_\rho)_{kj}\\
&=\sum_{k=1}^r(M_\rho)_{kj}
  \operatorname{Tr}_{K/F}(e_i e_k)\\
&=\operatorname{Tr}_{K/F}\!\left(
e_i\sum_{k=1}^r(M_\rho)_{kj}e_k
\right)\\
&=\operatorname{Tr}_{K/F}(e_i\,\rho e_j)\\
&=\operatorname{Tr}_{K/F}(\rho e_i e_j).
\end{aligned}$$
The Gram matrix of ${\mathcal B}_l$ in the basis
$\{e_1,\dots,e_r,\theta e_1,\dots,\theta e_r\}$
is
$$A=
\begin{pmatrix}
2G&0\\
0&-2GM_\rho
\end{pmatrix}.$$
Hence
$$\det A
= (-1)^r2^{2r}(\det G)^2\det M_\rho.$$
Since eigenvalues of $M_{\rho}$ are $\rho,\rho^q,\dots,\rho^{q^{r-1}}$, we have 
$$\det M_\rho 
=\rho^{1+q+\cdots+q^{r-1}}=N_{K/F}(\rho).$$
Then 
$$\det A
=(-1)^rN_{K/F}(\rho)\bigl(2^r\det G\bigr)^2.$$
Moreover, $N_{K/F}(\rho)$ is a nonsquare in $F$, since
\[
N_{K/F}(\rho)^{(q-1)/2}
=\rho^{(q^r-1)/2}=-1.
\]
All nonsquares in $F^*$ have the same square class. Thus
\[
\det A\equiv(-1)^r \eta\pmod{(F^*)^2},
\]
 where $\eta=N_{K/F}(\rho)$.
The type of Gram matrix is determined by $\det A$ modulo $(F^*)^2$, thus 
there exists a basis $\{w_1,\dots,w_l\}$ with Gram matrix $\text{diag}(1,1,\dots,1,(-1)^r\eta)$. 
If $q \equiv 3 \pmod{4}$ and $r$ is odd, then $(-1)^r\eta$ is a square in $\mathbb{F}_q$. Choose $c\in \mathbb{F}_q$ such that $c^2=-\eta$ and replace $w_l$ by $w'_l=c^{-1}w_l$, so ${\mathcal B}(w'_l,w'_l) = 1$. 
If $q \equiv 3 \pmod{4}$ and $r$ is even, then ${\mathcal B}(w_l,w_l) =\eta$.  
If $q \equiv 1 \pmod{4}$, then $(-1)$ is a square in $\mathbb{F}_q$. Choose $c\in \mathbb{F}_q$ such that $c^2= (-1)^r$ and replace $w_l$ by $w'_l=c^{-1}w_l$, so ${\mathcal B}(w'_l,w'_l) = \eta$. 
\end{proof}

 From the above proposition, equivalently, when $q$ is odd  and $l=2r$ is even, there exists a basis $\{w_i\}_{i=1}^l$ of $\mathbb{F}_{q^l}$ over $\mathbb{F}_q$ that satisfies the following
 \begin{equation}\label{basisl}
 \begin{split}
     \text{Tr}(w_i^{q^r}w_j)=0= \text{Tr}(w_iw_j^{q^r}) \text{ for } i\neq j, \\  \text{Tr}(w_i^{q^r+1})=1, \text{ for } \ 1\leq i\leq l-1 \text{ and Tr}(w_l^{q^r+1})=\epsilon,
      \end{split}
 \end{equation}
where $\epsilon$ is defined in Eq. (\ref{epsilon2}).
\begin{theorem}
    Let $q$ be odd and $l$ be even. Let $\{w_i\}_{i=1}^l$ be a basis for $\mathbb{F}_{q^l}$ over $\mathbb{F}_q$ satisfying Eq. \ref{basisl}. Then
 $$\langle g(x),f(x)\rangle_{th}=\langle g_1(x),f_1(x)\rangle_e+\dots+\langle g_{l-1}(x),f_{l-1}(x)\rangle_e+\epsilon \langle g_l(x),f_l(x)\rangle_e$$
 for all $g(x)=\sum_{i=1}^lw_ig_i(x)$ and $f(x)=\sum_{i=1}^lw_if_i(x)$. Moreover, in this case, we have  $$C_A^{\perp_{th}}=\left \{\sum_{i=1}^{l-1}w_if_i(x)+\epsilon^{-1}w_lf_l(x)\ | (f_1(x),\dots,f_l(x)\in C^{\perp_e}\right \},$$
where $C$ is a $\lambda$-quasi-twisted code of length $lm$ and index $l$ over $\mathbb{F}_q$ and $C_A$ is an additive $\lambda$-constacyclic code of length $m$ over $\mathbb{F}_{q^l}$ such that  $\phi(C)=C_A$.
\end{theorem}

 Let $\{w_i\}_{i=1}^l$ be a basis for $\mathbb{F}_{q^l}$ over $\mathbb{F}_q$ ($l=2r$). Let $f=(f_1(x),\dots,f_l(x)),g=(g_1(x),\dots,g_l(x))\in R^l$ and $f(x)=\sum_{i=1}^l w_if_i(x)$, $g(x)=\sum_{i=1}^l w_ig_i(x)$ be the corresponding elements in $R_A$. Then, the symplectic inner product is defined as
 $$\langle g,f\rangle_s=\langle g(x),f(x)\rangle_s=\sum_{i=1}^r\langle g_i(x),f_{r+i}(x)\rangle_e-\sum_{i=1}^r\langle g_{r+i}(x),f_i(x)\rangle_e.$$

Recall that when $q$ is even, $\mathcal{B}_l$ is a symplectic form on $\mathbb{F}_{q^l}$ over $\mathbb{F}_q$. The following is a well-known result for spaces with symplectic form (for instance, see  \cite[19.16, p. 81]{Aschbacher2000}).
\begin{proposition}\label{basisleven}
Let $q$ be even. Then there exists a basis  $\{w_i\}_{i=1}^l$ of $\mathbb{F}_{q^l}$ over $\mathbb{F}_q$ such that $\mathcal{B}_l(w_i,w_{r+i})=\mathcal{B}_l(w_{r+i},w_i)=1$ for $1\leq i\leq r$ and $\mathcal{B}_l(w_i,w_j)=0$ for other values of $i,j$.      
\end{proposition}

From the above proposition, for even $q$, there exists a basis  $\{w_i\}_{i=1}^l$  of $\mathbb{F}_{q^l}$ over $\mathbb{F}_q$ such that 
\begin{equation}\label{basisevenl}
 \begin{split}
     \text{Tr}(w_{r+i}w_i^{q^r})=\text{Tr}(w_iw_{r+i}^{q^r})=1 \text{ for } 1\leq i\leq r \text{ and } \text{Tr}(w_iw_j^{q^r})=0 \text{ for other values of } i,j.  
 \end{split}   
\end{equation}

\begin{theorem}
    Let $q$ be even. Let  $\{w_i\}_{i=1}^l$ be a basis of $\mathbb{F}_{q^l}$ over $\mathbb{F}_q$ satisfying Eq. \ref{basisevenl}. Then 
 $$\langle g(x),f(x)\rangle_{th} =\langle g(x),f(x)\rangle_s,$$ for all $g(x)=\sum_{i=1}^lw_ig_i(x)$ and $f(x)=\sum_{i=1}^lw_if_i(x)$.  Furthermore, in this case, we have  $$C_A^{\perp_{th}}=\left \{\sum_{i=1}^lw_if_i(x)\ |\  (f_1(x),f_2(x),\dots,f_l(x))\in C^{\perp_s}\right \},$$
where $C$ is a $\lambda$-quasi-twisted code of length $lm$ and index $l$ over $\mathbb{F}_q$ and $C_A$ is an additive $\lambda$-constacyclic code of length $m$ over $\mathbb{F}_{q^2}$ such that  $\phi(C)=C_A$.
\end{theorem}
\begin{proof}
By Proposition \ref{basisleven}, we have $\text{Tr}(w_{r+i}w_i^{q^r})=\text{Tr}(w_iw_{r+i}^{q^r})=1$ for $1\leq i\leq r$ and $\text{Tr}(w_iw_j^{q^r})=0$ for other values of $i,j$. Then,
\begin{equation*}
    \begin{split}
      \langle g(x),f(x)\rangle_{th}=&\sum_{i=1}^l\sum_{j=1}^l \text{Tr}(w_iw_j^{q^r})\langle g_i(x),f_j(x)\rangle_e\\
      =&\sum_{i=1}^r\langle g_i(x),f_{r+i}(x)\rangle_e+\sum_{i=1}^r\langle g_{r+i}(x),f_i(x)\rangle_e \\
      =&\langle g(x),f(x)\rangle_s.
    \end{split}
\end{equation*}
The rest of the proof is similar to Theorem \ref{qeventh}.
\end{proof}

The next example shows the existence of an additive cyclic code over $\mathbb{F}_4$ with parameters $(n,2^{2k},d)_4$, while there is no linear cyclic code with parameters $[n,k]_4$. We denote the parameters of additive codes over $\mathbb{F}_{q^2}$ by $(n,M,d)_{q^2}$. We denote the parameters of linear cyclic code over $\mathbb{F}_{q^2}$ by $[n,k,d]_{q^2}$.  To determine the best linear cyclic code, we constructed all possible cyclic codes for a given length and dimension using the MAGMA software. In Table \ref{tab3}, we provide several examples of additive $2$-constacyclic codes that outperform linear cyclic codes over $\mathbb{F}_9=\mathbb{F}_3[w]$ with $w^2+2w+2=0$.

\begin{example}
    Let $m=71$ and $\mathbb{F}_4=\mathbb{F}_2[w]$ with $w^2+w+1=0$. Let $C_A=\langle g_{11}(x)+wg_{12}(x),wg_{22}(x)\rangle$ be a additive cyclic code, where
    \begin{align*}
        g_{11}(x)=&x+1,\\
        g_{12}(x)=&x^{33} + x^{32} + x^{29} + x^{28} + x^{27} + x^{23} + x^{22} + x^{20} \\&+ x^{19} + x^{17}    + x^{14} + x^{12} + x^{11} + x^9 + x^7 + x^6 + x + 1\\
        g_{22}(x)=&x^{35} + x^{33} + x^{28} + x^{27} + x^{26} + x^{25} + x^{24} + x^{17} + x^{13} + x^8 + x^7 +
        x^5 + x^4 + x + 1.
    \end{align*}
 Then $C_A$ is  a $(71,2^{106},8)_4$ additive cyclic code over $\mathbb{F}_4$. Note that there is no linear cyclic code with parameters $[71,53]_4$.  
\end{example}

\begin{table}[h!]
    \centering
\small   \begin{tabular}{|c|c|c|c|c|c|c|}
    \hline
    $m$ & $g_{11}(x)$ &$g_{12}(x)$&$g_{22}(x)$&  \makecell{Additive\\ $\lambda$-constacyclic\\ code $C_A$} &\makecell{Best\\linear cyclic\\code}\\ 
    \hline
  $22$ & $x^2+1$   & $x^4 + x^3 + x$    & \makecell{$x^{10} + 2x^8 + 2x^6 $\\$+ 2x^4 + 1$}  & $(22,3^{32},5)_9$  &$[22,16,4]_9$\\
    \hline
 $21$  &  $x+1$  &  \makecell{$2x^{16} + 2x^{14} + 2x^{12} $\\$+ x^{11} + x^{10} + x^9 $\\$+ x^8 + 2x^7 + 2x^3 $\\$+
    x^2 + x + 1$}   &  \makecell{$(x+1)$\\$(x^6 + 2x^5 + x^4 $\\$+ 2x^3 + x^2 + 2x + 1)^3$}    & $(21,3^{22},8)_9$& $[21,11,5]_9$\\
    \hline
  $23$ & $x+1$    & $x^5 + 2x^4 + x^3 + 1$     & \makecell{$x^{11} + x^8 + x^6 $\\$+ 2x^4 + x^3 + x^2 $\\$+ 2x + 1$}    & $(23,3^{34},5)_9$  & \makecell{ No cyclic code\\ with parameters\\ $[23,17]_9$}\\
    \hline
  $25$ & $x+1$   &  \makecell{$x^{15} + 2x^{14} + 2x^{12} $\\$+ x^{11} + 2x^{10} + x^9 $\\$+ x^2 + 2*x + 1$}   & \makecell{$x^{21} + x^{20} + 2x^{16} $\\$+ 2x^{15} + x^{11} + x^{10} $\\$+ 2x^6 + 2x^5 + x + 1$}    & $(25,3^{28},8)_9$   & $[25,14,4]_9$\\
    \hline
$34$   & $x^2+1$   & \makecell{$x^{14} + x^{13} + 2x^{12} $\\$+ x^{10} + x^9 + x + 1$}    & \makecell{$x^{16} + x^{15} + 2x^{12} $\\$+ x^{11} + 2x^{10} + 2x^6 $\\$+ 2x^5 + 2x^4 + 2x + 1$}   &  $(34,3^{50},5)_9$& $[34,25,4]_9$\\
    \hline
    \end{tabular}
    \caption{ Examples of additive $2$-constacyclic codes that outperform linear cyclic codes over $\mathbb{F}_9=\mathbb{F}_3[w]$ with $w^2+2w+2=0$, where $C_A=\langle g_{11}(x)+wg_{12}(x),wg_{22}(x)\rangle$.}
    \label{tab3}
\end{table}

\section{Conclusion}\label{conclusion}
Quasi-twisted codes and additive constacyclic codes are broader classes of codes that generalize quasi-cyclic codes and additive cyclic codes, respectively. In this study, we investigate quasi-twisted codes of index $2$ over finite fields. We present a polynomial representation of these codes, analogous to that of quasi-cyclic codes, and derive their Euclidean, Hermitian, and symplectic duals. We also provide necessary and sufficient conditions for the self-orthogonality of these codes under appropriate choices of $\lambda$. Furthermore, we provide necessary and sufficient conditions for the inclusion of one quasi-twisted code within the Euclidean dual of another quasi-twisted code. 
Note that in Theorem \ref{gen1edaul}, Theorem \ref{gen1sdaul}, and Theorem \ref{gen1hdaul}, we determined the duals of one-generator quasi-twisted codes under the assumption $\gcd(m,q)=1$. This assumption is crucial for establishing the dimension arguments in the proofs, as $(x^m-\lambda)$, $(x^m-\lambda^{-1})$, and $(x^m-\lambda^{-q})$ have irreducible factors with multiplicity one. If $\gcd(m,q) \neq 1$, the polynomials possess repeated factors, and dimension arguments may not hold true in general. Therefore, studying this scenario would be an interesting research direction.

Next, we explore a one-to-one correspondence between quasi-twisted codes and additive constacyclic codes over extension fields, where the degree of the extension field is equal to the index of the quasi-twisted code. We establish relationships between trace Euclidean and trace Hermitian inner products in the additive setting and Euclidean and symplectic inner products in the quasi-twisted setting. Utilizing these relationships, we determine the trace Euclidean and trace Hermitian duals of additive constacyclic codes. Our findings reveal that the study of additive constacyclic codes with respect to trace inner products is equivalent to the study of quasi-twisted codes with respect to Euclidean and symplectic inner products.




\bibliographystyle{abbrv}
	\bibliography{ref}

@book{Aschbacher2000,
  title={Finite Group Theory},
  author={Michael Aschbacher},
  year={2000},
 publisher={Cambridge University Press},
  url={https://api.semanticscholar.org/CorpusID:118293616}
}

@book{Artin1957,
  title={Geometric Algebra},
  author={E. Artin},
  year={1957},
  publisher={Interscience Publishers, New York}
}

@article{Abdukhalikov20251,
  title={Quasi-cyclic codes of index 2},
  author={Kanat Abdukhalikov and  A. S. Dzhumadil’daev and San Ling},
  journal={Discrete Math.}, 
    FJOURNAL = {Discrete Mathematics},
    volume={349},
  year={2026}, 
      NUMBER = {6},
  pages={Paper No. 115004, 14},
  url={https://doi.org/10.1016/j.disc.2026.115004}
}

@article {Abdukhalikov2023,
    AUTHOR = {Abdukhalikov, Kanat and Bag, Tushar and Panario, Daniel},
     TITLE = {One-generator quasi-cyclic codes and their dual codes},
   JOURNAL = {Discrete Math.},
  FJOURNAL = {Discrete Mathematics},
    VOLUME = {346},
      YEAR = {2023},
    NUMBER = {6},
     PAGES = {Paper No. 113369, 13},
      ISSN = {0012-365X,1872-681X},
   MRCLASS = {94B15 (94B05)},
  MRNUMBER = {4548959},
MRREVIEWER = {Simon\ N.\ Litsyn},
       DOI = {10.1016/j.disc.2023.113369},
       URL = {https://doi.org/10.1016/j.disc.2023.113369},
}

@article {Seguin2004,
    AUTHOR = {S\'eguin, G\'erald E.},
     TITLE = {A class of 1-generator quasi-cyclic codes},
   JOURNAL = {IEEE Trans. Inform. Theory},
  FJOURNAL = {Institute of Electrical and Electronics Engineers.
              Transactions on Information Theory},
    VOLUME = {50},
      YEAR = {2004},
    NUMBER = {8},
     PAGES = {1745--1753},
      ISSN = {0018-9448,1557-9654},
   MRCLASS = {94B15},
  MRNUMBER = {2096842},
       DOI = {10.1109/TIT.2004.831861},
       URL = {https://doi.org/10.1109/TIT.2004.831861},
}

@article{ketkar2006,
  title={Nonbinary stabilizer codes over finite fields},
  author={Ketkar, Avanti and Klappenecker, Andreas and Kumar, Santosh and Sarvepalli, Pradeep Kiran},
  journal={IEEE transactions on information theory},
  volume={52},
  number={11},
  pages={4892--4914},
  year={2006},
  publisher={IEEE}
}

@article{daskalov2000,
  title={New quasi-twisted quaternary linear codes},
  author={Daskalov, Rumen N. and Gulliver, T Aaron},
  journal={IEEE Transactions on Information Theory},
  volume={46},
  number={7},
  pages={2642--2643},
  year={2000},
  publisher={IEEE}
}

@article{ackerman2011,
  title={New quinary linear codes from quasi-twisted codes and their duals},
  author={Ackerman, Ryan and Aydin, Nuh},
  journal={Applied mathematics letters},
  volume={24},
  number={4},
  pages={512--515},
  year={2011},
  publisher={Elsevier}
}

@techreport{prange1957,
  title={Cyclic Error-Correcting Codes in Two Symbols},
  author={Prange, E},
  year={1957},
  institution={TN-57-103}
}

@article{prange1985,
  title={Some cyclic error-correcting codes with simple decoding algorithms},
  author={Prange, Eugene},
  journal={AFCRC-TN-58-156},
  year={1985}
}

@article{aydin2017,
  title={Some results on the structure of constacyclic codes and new linear codes over {GF(7)} from quasi-twisted codes.},
  author={Aydin, Nuh and Connolly, Nicholas and Grassl, Markus},
  journal={Adv. Math. Commun.},
  volume={11},
  number={1},
  pages={245--258},
  year={2017}
}

@article{guneri2018,
  title={On the additive cyclic structure of quasi-cyclic codes},
  author={G{\"u}neri, Cem and {\"O}zdemir, Funda and Sole, Patrick},
  journal={Discrete Mathematics},
  volume={341},
  number={10},
  pages={2735--2741},
  year={2018},
  publisher={Elsevier}
}

@article{choi2023,
  title={Theory of additive complementary dual codes, constructions and computations},
  author={Choi, Whan-Hyuk and G{\"u}neri, Cem and Kim, Jon-Lark and {\"O}zbudak, Ferruh},
  journal={Finite Fields and Their Applications},
  volume={92},
  pages={102303},
  year={2023},
  publisher={Elsevier}
}

@article{huffman2007,
  title={Additive cyclic codes over $\mathbb{F}_4$},
  author={Huffman, W Cary},
  journal={Advances in Mathematics of Communications},
  volume={1},
  number={4},
  pages={427--459},
  year={2007},
  publisher={Advances in Mathematics of Communications}
}

@article{calderbank1998,
  title={Quantum error correction via codes over {GF(4)}},
  author={Calderbank, A Robert and Rains, Eric M and Shor, Peter M and Sloane, Neil JA},
  journal={IEEE Transactions on Information Theory},
  volume={44},
  number={4},
  pages={1369--1387},
  year={1998},
  publisher={IEEE}
}

@article{kasami1974,
  title={A {G}ilbert-{V}arshamov bound for quasi-cycle codes of rate $1/2$ (Corresp.)},
  author={Kasami, Tadao},
  journal={IEEE Transactions on Information Theory},
  volume={20},
  number={5},
  pages={679--679},
  year={1974},
  publisher={IEEE}
}

@article{lv2020,
  title={Constructions of quasi-twisted quantum codes},
  author={Lv, Jingjie and Li, Ruihu and Wang, Junli},
  journal={Quantum Information Processing},
  volume={19},
  pages={1--25},
  year={2020},
  publisher={Springer}
}

@inproceedings{saleh2024,
  title={Quantum Codes Derived from One-Generator Quasi-Twisted Codes},
  author={Saleh, Akram and Soleymani, Mohammad Reza},
  booktitle={2024 IEEE International Symposium on Information Theory (ISIT)},
  pages={2275--2280},
  year={2024},
  organization={IEEE}
}

@article{shi2016,
  title={Quasi-twisted codes with constacyclic constituent codes},
  author={Shi, Minjia and Zhang, Yiping},
  journal={Finite Fields and Their Applications},
  volume={39},
  pages={159--178},
  year={2016},
  publisher={Elsevier}
}

@article{aydin2001,
  title={The structure of 1-generator quasi-twisted codes and new linear codes},
  author={Aydin, Nuh and Siap, Irfan and Ray-Chaudhuri, Dijen K},
  journal={Designs, Codes and Cryptography},
  volume={24},
  number={3},
  pages={313--326},
  year={2001},
  publisher={Springer}
}

@article{qian2019,
  title={On self-dual and LCD quasi-twisted codes of index two over a special chain ring},
  author={Qian, Liqin and Shi, Minjia and Sol{\'e}, Patrick},
  journal={Cryptography and Communications},
  volume={11},
  pages={717--734},
  year={2019},
  publisher={Springer}
}

@article{wu2020,
  title={A modified {G}ilbert-{V}arshamov bound for self-dual quasi-twisted codes of index four},
  author={Wu, Rongsheng and Shi, Minjia},
  journal={Finite Fields and Their Applications},
  volume={62},
  pages={101627},
  year={2020},
  publisher={Elsevier}
}

@inproceedings{chepyzhov1993,
  title={A {G}ilbert-{V}ashamov bound for quasi-twisted codes of rate 1/n},
  author={Chepyzhov, V},
  booktitle={Proc. Joint Swedish-Russian Int. Workshop on Inf. Theory},
  pages={214--218},
  year={1993}
}

@article {daskalov2003,
    AUTHOR = {Daskalov, Rumen and Hristov, Plamen},
     TITLE = {New quasi-twisted degenerate ternary linear codes},
   JOURNAL = {IEEE Trans. Inform. Theory},
  FJOURNAL = {Institute of Electrical and Electronics Engineers.
              Transactions on Information Theory},
    VOLUME = {49},
      YEAR = {2003},
    NUMBER = {9},
     PAGES = {2259--2263},
      ISSN = {0018-9448,1557-9654},
   MRCLASS = {94B05 (94B65)},
  MRNUMBER = {2004781},
       DOI = {10.1109/TIT.2003.815798},
       URL = {https://doi.org/10.1109/TIT.2003.815798},
}

@article {jungickle1990,
    AUTHOR = {Jungnickel, Dieter and Menezes, Alfred J. and Vanstone, Scott
              A.},
     TITLE = {On the number of self-dual bases of {${\rm GF}(q^m)$} over
              {${\rm GF}(q)$}},
   JOURNAL = {Proc. Amer. Math. Soc.},
  FJOURNAL = {Proceedings of the American Mathematical Society},
    VOLUME = {109},
      YEAR = {1990},
    NUMBER = {1},
     PAGES = {23--29},
      ISSN = {0002-9939,1088-6826},
   MRCLASS = {11T30},
  MRNUMBER = {1007501},
MRREVIEWER = {Rudolf\ Lidl},
       DOI = {10.2307/2048357},
       URL = {https://doi.org/10.2307/2048357},
}

@article{shi2024,
  title={Additive Cyclic Codes Over $\mathbb{F}_{q^3}$},
  author={Minjia Shi and Shixiao Chu and Ferruh Ozbudak},
  journal={Journal of Algebra and Its Applications},
  year={2025},
  volume={24(9)},
  pages={2550339},
  url={https://api.semanticscholar.org/CorpusID:270804242}
}

@article {shi2022,
    AUTHOR = {Shi, Minjia and Liu, Na and \"Ozbudak, Ferruh and Sol\'e,
              Patrick},
     TITLE = {Additive cyclic complementary dual codes over $\mathbb{F}_4$},
   JOURNAL = {Finite Fields Appl.},
  FJOURNAL = {Finite Fields and their Applications},
    VOLUME = {83},
      YEAR = {2022},
     PAGES = {Paper No. 102087, 22},
      ISSN = {1071-5797,1090-2465},
   MRCLASS = {94B05 (12E10 15B05)},
  MRNUMBER = {4452043},
       DOI = {10.1016/j.ffa.2022.102087},
       URL = {https://doi.org/10.1016/j.ffa.2022.102087},
}

@article {reza2025,
    AUTHOR = {Dastbasteh, Reza and Shivji, Khalil},
     TITLE = {Polynomial representation of additive cyclic codes and new
              quantum codes},
   JOURNAL = {Adv. Math. Commun.},
  FJOURNAL = {Advances in Mathematics of Communications},
    VOLUME = {19},
      YEAR = {2025},
    NUMBER = {1},
     PAGES = {49--68},
      ISSN = {1930-5346,1930-5338},
   MRCLASS = {94B15 (94B05 94B65)},
  MRNUMBER = {4830124},
       DOI = {10.3934/amc.2023036},
       URL = {https://doi.org/10.3934/amc.2023036},
}

@article{ezerman2024,
  title={Characterization of Nearly Self-Orthogonal Quasi-Twisted Codes and Related Quantum Codes},
  author={Martianus Frederic Ezerman and Markus Grassl and San Ling and Ferruh {\"O}zbudak and Buket {\"O}zkaya},
  journal={IEEE transactions on information theory},
  volume={71},
  number={1},
  pages={499--517},
  year={2025},
  publisher={IEEE}
}

@article {jia2011,
    AUTHOR = {Jia, Yan},
     TITLE = {On quasi-twisted codes over finite fields},
   JOURNAL = {Finite Fields Appl.},
  FJOURNAL = {Finite Fields and their Applications},
    VOLUME = {18},
      YEAR = {2012},
    NUMBER = {2},
     PAGES = {237--257},
      ISSN = {1071-5797,1090-2465},
   MRCLASS = {94B05 (65T50 94B15)},
  MRNUMBER = {2890550},
MRREVIEWER = {Victor.\ A.\ Zinoviev},
       DOI = {10.1016/j.ffa.2011.08.001},
       URL = {https://doi.org/10.1016/j.ffa.2011.08.001},
}

@article {sangwisut2017,
    AUTHOR = {Sangwisut, Ekkasit and Jitman, Somphong and Udomkavanich,
              Patanee},
     TITLE = {Constacyclic and quasi-twisted {H}ermitian self-dual codes
              over finite fields},
   JOURNAL = {Adv. Math. Commun.},
  FJOURNAL = {Advances in Mathematics of Communications},
    VOLUME = {11},
      YEAR = {2017},
    NUMBER = {3},
     PAGES = {595--613},
      ISSN = {1930-5346,1930-5338},
   MRCLASS = {94B15},
  MRNUMBER = {3693862},
MRREVIEWER = {\'Angel\ del R\'io},
       DOI = {10.3934/amc.2017045},
       URL = {https://doi.org/10.3934/amc.2017045},
}

@book {Lang,
    AUTHOR = {Lang, Serge},
     TITLE = {Algebra},
    SERIES = {Graduate Texts in Mathematics},
    VOLUME = {211},
   EDITION = {third},
 PUBLISHER = {Springer-Verlag, New York},
      YEAR = {2002},
     PAGES = {xvi+914},
      ISBN = {0-387-95385-X},
   MRCLASS = {00A05 (15-02)},
  MRNUMBER = {1878556},
       DOI = {10.1007/978-1-4613-0041-0},
       URL = {https://doi-org.uaeu.idm.oclc.org/10.1007/978-1-4613-0041-0},
}

@article {seroussi1980,
    AUTHOR = {Seroussi, Gadiel and Lempel, Abraham},
     TITLE = {Factorization of symmetric matrices and trace-orthogonal bases
              in finite fields},
   JOURNAL = {SIAM J. Comput.},
  FJOURNAL = {SIAM Journal on Computing},
    VOLUME = {9},
      YEAR = {1980},
    NUMBER = {4},
     PAGES = {758--767},
      ISSN = {0097-5397},
   MRCLASS = {15A23 (15A33)},
  MRNUMBER = {592766},
MRREVIEWER = {Bradley W. Dickinson},
       DOI = {10.1137/0209059},
       URL = {https://doi.org/10.1137/0209059},
}

@incollection {lally2003,
    AUTHOR = {Lally, Kristine},
     TITLE = {Quasicyclic codes of index {$l$} over {${\bf F}_q$} viewed as
              {${\bf F}_q[x]$}-submodules of {${\bf F}_{q^l}[x]/\langle
              x^m-1\rangle$}},
 BOOKTITLE = {Applied algebra, algebraic algorithms and error-correcting
              codes ({T}oulouse, 2003)},
    SERIES = {Lecture Notes in Comput. Sci.},
    VOLUME = {2643},
     PAGES = {244--253},
 PUBLISHER = {Springer, Berlin},
      YEAR = {2003},
   MRCLASS = {94B10},
  MRNUMBER = {2043631},
MRREVIEWER = {Boris D. Kudryashov},
       DOI = {10.1007/3-540-44828-4\_26},
       URL = {https://doi.org/10.1007/3-540-44828-4_26},
}

@article {lally2001,
    AUTHOR = {Lally, Kristine and Fitzpatrick, Patrick},
     TITLE = {Algebraic structure of quasicyclic codes},
   JOURNAL = {Discrete Appl. Math.},
  FJOURNAL = {Discrete Applied Mathematics. The Journal of Combinatorial
              Algorithms, Informatics and Computational Sciences},
    VOLUME = {111},
      YEAR = {2001},
    NUMBER = {1-2},
     PAGES = {157--175},
      ISSN = {0166-218X},
   MRCLASS = {94B05 (13P10)},
  MRNUMBER = {1836725},
MRREVIEWER = {Edgar Mart\'{\i}nez-Moro},
       DOI = {10.1016/S0166-218X(00)00350-4},
       URL = {https://doi.org/10.1016/S0166-218X(00)00350-4},
}

@article{Verma2024,
  title={Trace dual of additive cyclic codes over finite fields},
  author={Gyanendra K. Verma and R. K. Sharma},
  journal={Cryptography and Communications},
   volume={16},
pages={ 1593-1608},
year={2024},
  url={https://link.springer.com/article/10.1007/s12095-024-00741-y}
}

@article{huffman2010,
  title={Cyclic $\mathbb{F}_q$-linear $\mathbb{F}_{q^t}$-codes},
  author={Huffman, W. C.},
  journal={ Int. J. Information and Coding Theory},
  volume={1},
  number={3},
  pages={249--284},
  year={2010}
}

@article {magma,
    AUTHOR = {Bosma, Wieb and Cannon, John and Playoust, Catherine},
     TITLE = {The {M}agma algebra system. {I}. {T}he user language},
      NOTE = {Computational algebra and number theory (London, 1993)},
   JOURNAL = {J. Symbolic Comput.},
  FJOURNAL = {Journal of Symbolic Computation},
    VOLUME = {24},
      YEAR = {1997},
    NUMBER = {3-4},
     PAGES = {235--265},
      ISSN = {0747-7171},
   MRCLASS = {68Q40},
  MRNUMBER = {MR1484478},
       DOI = {10.1006/jsco.1996.0125},
       URL = {http://dx.doi.org/10.1006/jsco.1996.0125},
}

@article{codetable,
	title = {Code Tables: Bounds on the parameters of various types of code},
	journal = {http://www.codetables.de (accessed on 20-12-2025)},
        doi = {http://www.codetables.de},
	url = { http://www.codetables.de},
	author = {Markus Grassl}
}

\end{document}